\documentclass[manuscript, natbib=false, USenglish, urlbreakonhyphens=false,nonacm]{acmart}
\usepackage[WOamssymb,WOtheoremdefinitions]{auxlib}
\usepackage{onlyamsmath}
\usepackage{xspace}


\usepackage{listings}
\usepackage{float}

\declaretheorem[style=plain]{theorem}

\declaretheorem[style=plain,sibling=theorem]{lemma}

\declaretheorem[style=remark,sibling=theorem]{claim}

\makeatletter
\def\@parfont{\bfseries}
\def\paragraph#1{\subsubsection*{#1.}}
\renewcommand\subsubsection{\@startsection{paragraph}{4}{\z@}%
  {-.5\baselineskip \@plus -2\p@ \@minus -.2\p@}%
  {-3.5\p@}%
  {\ACM@NRadjust{\@parfont}}}

\newif\ifsavetrees\savetreestrue
\ifsavetrees
\renewcommand\subsubsection{\@startsection{paragraph}{4}{\z@}%
  {-.25\baselineskip \@plus -1\p@ \@minus -.1\p@}%
  {-1.75\p@}%
  {\ACM@NRadjust{\@parfont}}}

\let\markeverypar\everypar
\newtoks\everypar
\everypar\markeverypar
\markeverypar{\the\everypar\looseness=-1\relax}

\fi
\makeatother

\usepackage[
        isbn=false,
        sortcites,
        maxbibnames=99,
        maxcitenames=2,
        maxcitenames=10,
        maxalphanames=4,
        minalphanames=3,
        style=numeric,
        url=false,
        giveninits=true,
        backend=biber,
    ]{biblatex}

\AtEveryBibitem{\ifentrytype{inproceedings}{\clearname{editor}\clearlist{address}\clearlist{location}}{}}
\AtEveryBibitem{\clearlist{language}}

\createcustomnote{todo}{backgroundcolor=AUXLorange}

\makeatletter
\newcommand{\DeclareMathActive}[2]{%
  \expandafter\edef\csname keep@#1@code\endcsname{\mathchar\the\mathcode`#1 }
  \begingroup\lccode`~=`#1\relax
  \lowercase{\endgroup\def~}{#2}%
  \AtBeginDocument{\mathcode`#1="8000 }%
}

\newcommand{\std}[1]{\csname keep@#1@code\endcsname}
\patchcmd{\newmcodes@}{\mathcode`\-\relax}{\std@minuscode\relax}{}{\ddt}
\AtBeginDocument{\edef\std@minuscode{\the\mathcode`-}}
\makeatother

\DeclareMathActive{O}{\operatorname{\std{O}}}

\def\cO{\mathcal{\std{O}}}

\floatname{algoflt}{Algorithm}
\floatstyle{boxed}
\newfloat{algoflt}{tbp}{loa}
\def\dcmnumberstyle{}
\def\dcmbasicstyle{}
\makeatletter
\lst@AddToHook{OnEmptyLine}{\addtocounter{lstnumber}{-1}\vspace{-0.5\baselineskip}\global\def\dcmnumberstyle{\color{white}\global\def\dcmnumberstyle{}}}%

\newif\ifdcmlinenumbers
\dcmlinenumberstrue
\lstnewenvironment{lstalgo}[2][htb]{%
\lstset{%
	mathescape,%
	tabsize=2,%
	numbers=left,%
	numberstyle=\tiny\dcmnumberstyle,%
	numberblanklines=true,%
	showlines=true,%
	basicstyle=\linespread{0.8}\rmfamily\dcmbasicstyle,%
	columns=fullflexible,%
	emph={%
		[1]if, else, then, and, or, for, do, once, while, return, exit, not, output, with %
		},
	emphstyle={%
		[1]\bfseries%
		},%
	escapeinside={/*}{*/},
	escapebegin={\ifmmode\else\color{gray}\hfill\small$\triangleright$ \fi},
	escapeend={\ifmmode\else\fi}
	}%
	\ifdcmlinenumbers
	\else
	\lstset{numbers=none}
	\fi
\setlength{\topsep}{0pt}%
\setlength{\parindent}{0pt}%
\@skiphyperreftrue%
\lst@CCPutMacro
    \lst@ProcessOther {"2D}{\lst@ttfamily{-{}}{-}}
    \@empty\z@\@empty
\algoflt[#1]\caption{#2}%
}{\endalgoflt\@skiphyperreffalse}

\AfterPackageLoaded{cleveref}{
\crefname{lst@lineno}{Line}{Lines}
\crefname{algoflt}{Algorithm}{Algorithms}
}

\makeatother

\def\Varname#1{\texttt{#1}}

\def\Vardef#1{%
	\expandafter\newcommand\csname #1\endcsname[1]{%
		\def\first{##1}%
		\def\second{*}%
		\def\third{}%
		\ensuremath{\Varname{\MakeLowercase #1}\ifx\first\second\else\ifx\first\third\else[{##1}]\fi\fi}%
	}%
}
\def\Vardefxx#1#2{%
	\expandafter\newcommand\csname #1\endcsname[1]{%
		\def\first{##1}%
		\def\second{*}%
		\def\third{}%
		\ensuremath{\Varname{#2}\ifx\first\second\else\ifx\first\third\else[{##1}]\fi\fi}%
	}%
}
\def\Vardefx#1#2{%
	\expandafter\newcommand\csname #1\endcsname[1]{%
		\ensuremath{\Varname{#2}[{##1}]}%
	}%
}

\Vardef{Rounds}
\Vardef{Tick}
\Vardefx{DecisionFlag}{decision}
\Vardefx{TickE}{tick_{E}}
\Vardef{Undecided}

\Vardef{PhaseType}
\Vardef{Round}

\Vardef{Part}
\Vardef{Subphase}
\Vardef{Phase}
\Vardef{Opinion}
\Vardef{Weight}

\Vardef{Role}
\Vardef{Count}
\Vardef{Counter}
\Vardef{Ready}
\Vardef{Champion}
\Vardef{Defender}
\Vardef{Challenger}
\Vardef{Winner}
\Vardef{Playeropinion}
\Vardef{Leader}
\Vardef{Junta}
\Vardef{Level}
\Vardef{Active}
\Vardef{Current}
\Vardef{Informing}
\Vardef{Candidate}
\Vardef{Candidateopinion}
\Vardef{Selectedcandidate}
\Vardef{Remainingcandidate}
\Vardefxx{Tc}{tcnt}
\Vardef{Tokens}
\Vardef{Hour}
\Vardef{Minute}
\Vardef{Finished}
\Vardef{Output}
\Vardef{Input}
\Vardef{Bias}
\Vardef{End}

\def\KWdef#1{\expandafter\def\csname #1\endcsname{\ensuremath{\Varname{\MakeLowercase #1}}\xspace}}

\KWdef{Collector}
\KWdef{Player}
\KWdef{Tracker}
\KWdef{Setup}
\KWdef{Match}
\KWdef{Conclusion}
\KWdef{Clock}
\KWdef{Broadcast}
\KWdef{Init}
\KWdef{Cleanup}
\KWdef{Majority}

\KWdef{Cancellation}
\KWdef{Lineup}
\KWdef{Maj}

\usepackage{xspace}
\def\ttrue{\ensuremath{\text{\textsc{true}}}\xspace}
\def\tfalse{\ensuremath{\text{\textsc{false}}}\xspace}

\usepackage{tikz}

\makeatletter
\def\?#1{}
\def\whp{w.h.p\@ifnextchar.{.\?}{\@ifnextchar,{.}{\@ifnextchar:{.}{\@ifnextchar){.}{.\ }}}}}
\def\Whp{W.h.p\@ifnextchar.{.\?}{\@ifnextchar,{.}{.\ }}}
\makeatother

\def\protocol#1{\textsc{#1}\xspace}
\def\AlgSimple{\protocol{SimpleAlgorithm}}
\def\AlgComplex{\protocol{ImprovedAlgorithm}}

\def\FormJunta{\protocol{FormJunta}}

\title{Population Protocols for Exact Plurality Consensus}

\subtitle{How a small chance of failure helps to eliminate insignificant opinions}

\author{Gregor Bankhamer}
\email{gbank@cs.sbg.ac.at}
\affiliation{%
    \institution{University of Salzburg}
    \city{Salzburg}
    \country{Austria}
}
\author{Petra Berenbrink}
\email{petra.berenbrink@uni-hamburg.de}
\affiliation{%
    \institution{Universität Hamburg}
    \city{Hamburg}
    \country{Germany}
}
\author{Felix Biermeier}
\email{felix.biermeier@uni-hamburg.de}
\affiliation{%
    \institution{Universität Hamburg}
    \city{Hamburg}
    \country{Germany}
}
\author{Robert Els{\"a}sser}
\email{elsa@cs.sbg.ac.at}
\affiliation{%
    \institution{University of Salzburg}
    \city{Salzburg}
    \country{Austria}
}
\author{Hamed Hosseinpour}
\email{hamed.hosseinpour@uni-hamburg.de}
\affiliation{%
    \institution{Universität Hamburg}
    \city{Hamburg}
    \country{Germany}
}
\author{Dominik Kaaser}
\email{dominik.kaaser@tuhh.de}
\affiliation{%
    \institution{TU Hamburg}
    \city{Hamburg}
    \country{Germany}
}
\author{Peter Kling}
\email{peter.kling@uni-hamburg.de}
\orcid{0000-0003-0000-8689}
\affiliation{%
    \institution{Universität Hamburg}
    \city{Hamburg}
    \country{Germany}
}

\begin{document}
	\newcommand*{\abstractcite}[1]{[\citeauthor{#1}; \citeyear{#1}]}

\begin{abstract}
We consider the \emph{plurality consensus} problem for \emph{population protocols}.
Here, $n$ anonymous agents start each with one of $k$ opinions.
Their goal is to agree on the initially most frequent opinion (the \emph{plurality opinion}) via random, pairwise interactions.
\emph{Exact} plurality consensus refers to the requirement that the plurality opinion must be identified even if the \emph{bias} (difference between the most and second most frequent opinion) is only $1$.

The case of $k = 2$ opinions is known as the \emph{majority problem}.
Recent breakthroughs led to an always correct, exact majority population protocol that is both time- and space-optimal, needing $O(\log n)$ states per agent and, with high probability, $O(\log n)$ time~\abstractcite{DBLP:journals/corr/abs-2106-10201}.
Meanwhile, results for general plurality consensus are rare and far from optimal.
We know that any always correct protocol requires $\Omega(k^2)$ states, while the currently best protocol needs $O(k^{11})$ states~\abstractcite{DBLP:conf/ciac/NataleR19}.
For ordered opinions, this can be improved to $O(k^6)$~\abstractcite{DBLP:conf/opodis/GasieniecHMSS16}.

We design protocols for plurality consensus that beat the quadratic lower bound by allowing a negligible failure probability.
While our protocols might fail, they identify the plurality opinion with high probability even if the bias is $1$.
Our first protocol achieves this via $k-1$ tournaments in time $O(k \cdot \log n)$ using $O(k + \log n)$ states.
While it assumes an ordering on the opinions, we remove this restriction in our second protocol, at the cost of a slightly increased time $O(k \cdot \log n + \log^2 n)$.
By efficiently pruning insignificant opinions, our final protocol reduces the number of tournaments at the cost of a slightly increased state complexity $O(k \cdot \log\log n + \log n)$.
This improves the time to $O(n / x_{\max} \cdot \log n + \log^2 n)$, where $x_{\max}$ is the initial size of the plurality.
Note that $n/x_{\max}$ is at most $k$ and can be much smaller (e.g., in case of a large bias or if there are many small opinions).
\end{abstract}

	\maketitle
	\section{Introduction}%
\label{sec:introduction}

In this paper, we design and analyze a new population protocol for \emph{plurality consensus}, a fundamental problem in distributed computing.
There are $n$ \emph{agents}, each starting with one of $k$ \emph{opinions}.
The goal is that all agents eventually agree on the initially most frequent opinion.

The \emph{population protocol model}~\cite{DBLP:journals/dc/AngluinADFP06} has become a popular way to study distributed systems formed by many simple, resource-limited agents.
A key feature of the model is that communication is erratic: agents cannot choose their communication partners, but instead each time step one (typically random) pair of agents is chosen to interact.
During such an \emph{interaction}, both agents observe each other's state and use a common transition function to update their respective states.
The random communication and simple state updates make the model particularly suitable for systems of many, simple entities whose communication patterns seem unpredictable, like chemical reactions~\cite{DBLP:journals/nc/SoloveichikCWB08, DBLP:conf/soda/Doty14, DBLP:journals/dc/ChenCDS17}, gene regulatory networks~\cite{ISBN:9780262524230}, animal populations~\cite{DOI:10.1093/icesjms/3.1.3}, or opinion formation in social groups~\cite{10.1007/bf03160228}.

Population protocols are compared with respect to their \emph{space complexity} (measured in states per agent) and \emph{how many interactions} they require to reach and stay in a desired \emph{configuration} (a global system state, like all agents agreeing on one opinion or one agent being in a leader state).
We express time bounds in the standard notion \emph{parallel time}, which is the number of interactions divided by $n$.
Thus, in expectation each agent takes part in $\ldauTheta{1}$ interactions per time unit.

The original model~\cite{DBLP:journals/dc/AngluinADFP06} restricts the number of states per agent to a constant with respect to the population size~$n$.
The computational power and limits of such constant-state population protocols are well understood~\cite{DBLP:journals/dc/AngluinAER07}, at least for \emph{stable} protocols (which must \emph{always}, with probability $1$, reach and stay in a desired configuration).
While the picture is less clear for protocols whose state space grows with $n$, recent breakthrough results managed to design stable protocols for \enquote{benchmark} problems like \emph{leader election} and \emph{majority} that are simultaneously time- and space-optimal~\cite{DBLP:conf/stoc/BerenbrinkGK20, DBLP:journals/corr/abs-2106-10201}.

\paragraph{Plurality Consensus in Population Protocols}
In the following, \emph{plurality opinion} refers to the opinion with the initially largest support (assuming it is unique) and \emph{bias} denotes the difference of that opinion's initial support to that of the second largest opinion.
A major part of research seeks to identify this plurality opinion for \emph{any} initial bias, even if it is only $1$.
This is often referred to as the \emph{exact} plurality problem.
In contrast, the \emph{approximate} plurality problem requires to identify the plurality opinion only if the initial bias is large enough (typically of order $\ldauomega{\sqrt{n}}$).

The \emph{majority problem} is a special case of plurality consensus, considering only $k = 2$ initial opinions.
As a fundamental problem in distributed computing, a lot of work has been invested to find an (asymptotically) optimal, stable population protocol for exact majority~\cite{DBLP:conf/podc/AlistarhGV15, DBLP:conf/soda/AlistarhAG18, DBLP:conf/wdag/BerenbrinkEFKKR18, DBLP:conf/podc/NunKKP20, DBLP:journals/dc/BerenbrinkEFKKR21}, culminating in~\cite{DBLP:journals/corr/abs-2106-10201}, which solves majority using both $\ldauOmicron{\log n}$ states and expected parallel time.
This is optimal, in that no stable protocol can solve majority faster ($\ldauOmega{\log n}$ is the time until each agent interacted at least once) and any polylogarithmic-time stable majority protocol requires $\ldauOmega{\log n}$ states (under two natural conditions, see~\cite{DBLP:conf/soda/AlistarhAG18}).
Note that the difficulty here stems from requiring \emph{exactness}; for \emph{approximate} majority, a simple 3-state protocol identifies the majority \whp in parallel time $\ldauOmicron{\log n}$ if the initial bias is $\ldauOmega{\sqrt{n \log n}}$~\cite{DBLP:journals/dc/AngluinAE08, DBLP:journals/nc/CondonHKM20}.
Focusing on constant-state protocols that might fail, \cite{conf/podc/KosowskiU18} mentions a protocol with constant state space and which \whp determines the exact majority in time $\ldauOmicron{\log^3 n}$.

Population protocols for general plurality consensus are scarce.
One line of research studies the state complexity (ignoring time) required to \emph{always} identify the plurality opinion.
While one needs at least $k$ states to represent $k$ opinions, \textcite{DBLP:conf/ciac/NataleR19} show that always correct plurality consensus requires even $\ldauOmega{k^2}$ states.
The currently best always correct protocol needs $\ldauOmicron{k^{11}}$ states, which can be reduced to $\ldauOmicron{k^6}$ if there is a total ordering on the opinions~\cite{DBLP:conf/opodis/GasieniecHMSS16}.
The quadratic lower bound makes it apparent that \emph{always} guaranteeing a correctly identified plurality opinion comes at the cost of high space complexity.
Sacrificing these strong guarantees, \cite{10.1137/1.9781611977073.135} reaches consensus \whp in $\ldauOmicron{\log^2 n}$ parallel time using only $k \cdot \ldauOmicron{\log n}$ states.
However, they only consider \emph{approximate} plurality consensus, requiring an initial bias of order $\ldauOmega{\sqrt{n \log n}}$.

\paragraph{Our Results in a Nutshell}
We present new population protocols for plurality consensus.
Our protocols may fail with negligible probability, allowing us to beat the quadratic lower bound on the state space~\cite{DBLP:conf/ciac/NataleR19} (our bounds are almost linear in $k$).
In contrast to~\cite{10.1137/1.9781611977073.135} our protocols remain exact: w.h.p.\footnote{The expression \emph{with high probability (w.h.p.)} refers to a probability of $1 - n^{-\LDAUOmega{1}}$.} they identify the plurality opinion, even if the initial bias is $1$.
Our first protocol relies on ordered opinions to eliminate non-plurality opinions in $k-1$ tournaments.
The second protocol works similarly but removes the need for ordered opinions at the cost of a slightly larger runtime.
Our final protocol allows for a significant speedup by quickly removing insignificant opinions before the tournaments start.

	\section{Model and Results}%
\label{sec:model_and_results}

We consider a system of $n$ identical, anonymous \emph{agents} (finite state machines) with state space $Q$ (whose size may depend on $n$).
In every time step, one pair of agents $(u, v)$ is chosen independently and uniformly at random to interact.
During such an \emph{interaction}, both agents update their states according to a common transition function $\delta\colon Q \times Q \to Q \times Q$.

In the \emph{plurality consensus} problem, each agent starts with one opinion out of a set $\cO$ of $k$ \emph{opinions}, which may or may not be totally ordered.
If $\cO$ is ordered, we assume w.l.o.g.\ that $\cO = \set{1, 2, \dots, k}$.
The \emph{bias} is the difference between the support of the most and second most frequent opinion.
Assuming that the initial bias is at least $1$, we call the initially most frequent opinion the \emph{plurality opinion}.
The goal is for all agents to output the initial plurality opinion.

We represent the initial \emph{distribution of opinions} as a vector $\bm{x} = \intoo{x_i}_{i \in \cO}$, where $x_i$ denotes the number of agents that initially have opinion $i$.
Additionally, $x_{\max} = \max_{i \in \cO}\set{x_i}$ is the initial support of the plurality opinion.
A \emph{configuration} describes the global system state at a given time (e.g., by stating how many agents are in each possible state).

\smallskip

We design protocols that, \whp, identify the plurality opinion quickly and have an almost optimal space complexity, even if the initial bias is only $1$ (hence we solve \emph{exact} plurality consensus).
With this goal, allowing a negligible failure probability is essential, as otherwise – independently of the runtime – any protocol requires $\ldauOmega{k^2}$ states~\cite{DBLP:conf/ciac/NataleR19}.

Our first protocol uses $\ldauOmicron{k + \log n}$ states.
It consists of $k-1$ tournaments, during each of which a \emph{defender} and \emph{challenger} opinion compete.
\Whp, the plurality opinion emerges victorious from all tournaments in time $\ldauOmicron{k \cdot \log n}$.
This protocol relies on an ordering of the opinions to determine the next challenger opinion.
Our second protocol avoids the requirement of such an order by using instead a \emph{leader election} subprotocol to determine the next challenger opinion.
Using the leader election protocol from~\cite{DBLP:journals/jacm/GasieniecS21} for this,\footnote{%
    \Whp, that protocol finishes in $\ldauOmicron{\log^2 n}$ time, leading to the corresponding term in our increased runtime.
    While there is a $\ldauOmicron{\log n}$ time leader election protocol~\cite{DBLP:conf/stoc/BerenbrinkGK20}, that runtime holds only in expectation, which is too weak for our purpose.
} our protocol for unordered opinions still uses $\ldauOmicron{k + \log n}$ states but has a slightly increased runtime of $\ldauOmicron{k \cdot \log n + \log^2 n}$.
By itself, avoiding such an ordering might seem like an esoteric challenge, but this approach plays a crucial role in our third protocol (see below), where it is used to perform tournaments only for a subset of a priori unknown (significant) opinions that remain after an initial pruning phase. The following theorem states the properties of our first two protocols.
\begin{theorem}%
\label{thm:result-order}
Assume we have a population of size $n$ with $k \leq n/40$ initial opinions.
\begin{enumerate}[nosep]
\item If the opinions are numbered $1, \dots, k$ then \AlgSimple converges \whp to the plurality opinion in $\ldauOmicron{k \cdot \log n}$ parallel time using $\ldauOmicron{k + \log n}$ states.
\item If there is no order among the opinions, \AlgSimple can be modified to converge \whp to the initial plurality opinion in $\ldauOmicron{k \cdot \log n + \log^2 n}$ parallel time using $O(k + \log n)$ states.
\end{enumerate}
\end{theorem}
Note, by tightening the analysis and slightly modifying  
\cref{alg:synchronization} our \cref{thm:result-order} also holds for $k \le n-1$. These extensions are given in \cref{sec:appendix-larger-k}. For constant values of $k$, the unmodified \AlgSimple converges \whp in optimal $\ldauOmicron{\log n}$ parallel time and requires only $\ldauOmicron{\log n}$ states. 
This matches state and time complexities of the state-of-the-art exact majority protocol \cite{DBLP:conf/focs/DotyEGSUS21}. Note that the protocol from \cite{DBLP:conf/focs/DotyEGSUS21} is stable but ours gives \whp guarantees only.
A detailed description and analysis of the \lcnamecref{thm:result-order}'s first statement can be found in \cref{sec:algorithm}.
Details for the \lcnamecref{thm:result-order}'s second statement are given in \cref{sec:removing-the-order}.

Our main contribution is the third protocol, which uses a pruning process to remove \emph{insignificant} opinions before the tournaments start, reducing their number from $k-1$ to $n / x_{\max}$ (remember that $x_{\max}$ denotes the initial size of the plurality opinion).
An opinion $j$ is called insignificant if $x_j \leq x_{\max} / c_s$, where $c_s > 1$ is a suitable constant.
If $x_{\max}$ is of order $n^{1/2 + \ldauOmega{1}}$, the resulting protocol \whp identifies the plurality opinion in parallel time $\ldauOmicron{n / x_{\max} \cdot \log n + \log^2 n}$, using a slightly larger state space of size $\ldauOmicron{k \cdot \log\log n + \log n}$.
Note that $n / x_{\max}$ is always at most $k$ but it may be much smaller (e.g., if one opinion is very large or if there are many small opinions).

The idea of the pruning process is to have each \emph{subpopulation of opinions} run through a few \emph{preprocessing phases} controlled by their own, dedicated \emph{phase clock}. Phase clocks \cite{DBLP:conf/soda/AlistarhAG18,DBLP:journals/dc/AngluinAE08a,DBLP:journals/dc/BerenbrinkEFKKR21,DBLP:journals/jacm/GasieniecS21} are a common tool in population protocols to synchronize agents into phases.
We will show that larger subpopulations finish their preprocessing phases earlier than smaller subpopulations.
Once the first subpopulation finished their preprocessing phases, we use a broadcast to prune any subpopulation (opinion) whose agents have not progressed far enough and then start the actual tournaments with the remaining opinions. Since we cannot know which and how many opinions remain after this pruning phase, we rely on the approach of our second protocol to select the next challenger opinion via a leader election subprotocol.

The following \lcnamecref{thm:result-no-order-phase-clock} states the results for our final population protocol formally.
\begin{theorem}%
\label{thm:result-no-order-phase-clock}
Assume we have a population of size $n$ with $k$ initial opinions where $x_{\max} > n^{1/2 + \varepsilon}$ for some small constant $1/2 > \varepsilon > 0$.
\AlgComplex converges \whp to the plurality opinion in $\ldauOmicron{ {n}/{x_{\max}} \cdot \log n + \log^2 n}$ parallel time using $\ldauOmicron{k \cdot \log\log n+ \log n}$ states.
\end{theorem}
A description of the corresponding protocol and its analysis is given in \cref{sec:filtering}.
Note that if $k < n^{1/2 - \varepsilon}$, the requirement $x_{\max} > n^{1/2 + \varepsilon}$ is always fulfilled (this follows from $x_{\max} \geq n/k$).

	\section{The Simple Algorithm}%
\label{sec:algorithm}

In this section we present our first algorithm called \AlgSimple where each agent has one of $k$ possible opinions numbered from $1$ to $k$. 
The main idea of the protocol is as follows.
It performs a sequence of \emph{tournaments} of length $\ldauOmicron{\log n}$ synchronized by a \emph{phase clock} \cite{DBLP:conf/soda/AlistarhAG18}.
In each tournament two fixed opinions are chosen, and an exact majority protocol~\cite{DBLP:conf/focs/DotyEGSUS21} is used to determine the majority opinion among the two of them.
In the first tournament opinions $1$ and $2$ compete.
In tournament $i > 1$ the winner of tournament $i-1$ (called \emph{defender}) competes against opinion $i+1$ (called \emph{challenger}).
The winner of tournament $i$ has the largest support among the first $i+1$ opinions, and the winner of the last tournament is the plurality opinion. 

In order to reach our state bound of $\ldauOmicron{k+\log n}$, our protocol has to be very economical with the states.
For example, it is not possible for an agent to store two different opinions which would already require $\ldauOmega{k^2}$ states.
Our protocol starts with an initialization phase which splits the agents into four \emph{roles}: \Collector, \Player, \Clock, and \Tracker.
Every agent $u$ has a variable $\Role u$ to store its role in the protocol.
The protocol consists of an initialization part (see \cref{alg:init}) and three different subprotocols that are specific to the corresponding roles.

We already argued that every agent cannot store two different opinions.
Hence, the initialization phase is used to \enquote{collect} opinions: Initially each agent is a $\Collector$-agent for its initial opinion.
Each agent has a variable $\Tokens{}$ which can take on values between $1$ and $10$.
For every opinion, the total number of tokens of that opinion equals the number of agents initially supporting that opinion.
When a $\Collector$-agent meets another agent with the same opinion it increases the token counter accordingly.
This frees up the other agent which takes on a role in $\set{\Clock,\Tracker,\Player}$. 
During the tournament the $\Collector$-agents have the responsibility to initiate the majority protocols between the actual challenger and defender. 
To this end they have two Boolean variables \Defender{} and \Challenger{} which indicate that their opinion participates in the match as defender or challenger, respectively.
Additionally, all \Collector-agents have a bit \Winner{} which indicates the majority opinion of the last tournament. This bit is used to broadcast the final majority opinion.
Finally, a value $\ell \in [-10, 10]$ is used to cancel opposing opinions before a match.

Internally the $\Clock{}$ agents run the leaderless phase clock from \cite{DBLP:conf/soda/AlistarhAG18} on a local counter $\Count {}$ (see \cref{sec:clock-synchronization}).
Whenever the local counter passes through zero the agent increases a variable $\Phase{}$ modulo 10.
The new value is disseminated to all other agents via one-way epidemics.
The role of the $\Tracker$-agents is to store the number of the current challenger in a variable $\Tc{}$ (short for \emph{tournament counter}).
Whenever one of the tournaments is over this variable is increased by one.
This is used by the $\Collector$-agents to set the challenger bit at the beginning of a new tournament.
The $\Player$-agents are the ones performing the $k-1$ tournaments.
At the beginning of a tournament these agents adopt the opinions from $\Collector$-agents which have either the \emph{defender} or \emph{challenger} bit set and set their $\Playeropinion{}$ to $A$ or $B$, respectively.

\paragraph{Overview of the State Space}
We use $\cS_{\Maj}$ to denote the set of states used by the exact majority protocol.
\Cref{fig:state-space} gives an overview of the variables used by our protocol and how some of them can be attributed to the different roles.
\begin{figure*}
\fbox{\begin{minipage}{\textwidth}
{\setlength{\fboxsep}{2pt}
\newcommand*{\graybox}[2]{\raisebox{0.5em}{\colorbox{gray}{\parbox{#1}{\centering\color{white}#2}}}}
\[
\cS = 
\arraycolsep=0.5pt
\begin{array}{@{}ccccccccccccccccccc@{}}
 & & & & \graybox{4em}{\mathstrut\Clock} & & \graybox{4.1em}{\mathstrut\Tracker} & & \multicolumn{7}{c}{\graybox{14em}{\mathstrut\Collector}} & & \multicolumn{3}{c}{\graybox{7em}{\mathstrut\Player}} \\
\underbrace{\phantomas{\set{\Collector,}}{\raisebox{1em}{$\Set{\substack{\Collector,\\ \Player,\\ \Clock,\\ \Tracker}}$}}}_{\Role{}} & \times &
\underbrace{\intcc{-1;9}}_{\Phase{}} & \times & %
\underbrace{[\Theta(\log n)]}_{\Count{}} & \times &
\underbrace{[k]}_{\Tc{}} & \times &
\underbrace{[k]}_{\mathclap{\Opinion{}}} & \times &
\underbrace{[10]}_{\mathclap{\Tokens{}}} & \times &
\underbrace{\set{0, 1}^3}_{\mathclap{\substack{\Defender{},\\\Challenger{},\\\Winner{}}}} & \times &
\underbrace{[-10;10]}_{\ell} & \times &
\underbrace{\set{A,B,U}}_{\mathclap{\Playeropinion{}}} &  \times &
\cS_{\Maj}\end{array}
\]}
\end{minipage}}\vspace{-2ex}
\caption{State Space $\cS$. Note that $[i] = \set{1, \dots, i}$ and $[-i;j] = \set{-i, \dots, j}$.}\label{fig:state-space}
\end{figure*}
Note that $\cS$ is \emph{not} the actual state space used by our protocol.
Our actual state space is much smaller, since the role-specific variables must only be kept track of by the corresponding roles.
We describe this more thoroughly in the corresponding proof of the state complexity in \cref{sec:proof-sketch}.

\paragraph{Simplifications for the Pseudocode}
In our formal algorithms we define how both involved agents $(u,v)$ update their states in an interaction: $u$ is the initiator and $v$ is the responder of that interaction.
To simplify the exposition of our protocols, we allow the use of a "do once" statement in the pseudocode for state transitions that are to be executed \emph{only once} in a given phase.
For example, consider the scenario where the challenger wins the match.
In the subsequent conclusion phase, all defender agents remove and all challenger agents set the $\Defender{}$ bit.
This must be done exactly once, since otherwise all bits are lost.
See \cref{ln:doonceexample:start} to \cref{ln:doonceexample:end} in \cref{alg:setup} for the corresponding pseudocode using a \enquote{do once} statement.
Similarly to the \enquote{do once} statements we assume that agents can determine whether they interact for the \enquote{first time} in a phase.
Note that these statements can be implemented using constantly many bits, such that the overall state space size increases only by a constant factor.

\paragraph{Outline}
In \cref{sec:clock-synchronization} we first describe the protocol for the $\Clock{}$- and $\Tracker{}$-agents. 
In \cref{sec:init} we describe the initialization routine in more detail and present an analysis for the initialization phase.
In \cref{sec:tournament} we present a formal definition of the protocol used by the \Collector-agents and \Player-agents.
A sketch of the proof of the first statement of \cref{thm:result-order} can be found in \cref{sec:proof-sketch} and the complete proof can be found in \cref{sec:proof-theorem1}.

\subsection{Clock and Tracker Agents} \label{sec:clock-synchronization}

The $\Clock$-agents have two different tasks (see \cref{alg:synchronization}).
First they decide when the initialization phase is over.
For that they use their local counter $\Count{}$ (initialized to zero).
Whenever they interact with a non-\Collector-agent they increase $\Count{}$ by one.
If they interact with a \Collector-agent $\Count {}$ is decreased by one as long as it is larger that zero.
As soon as $\Count {}$ reaches $5\log n$ the agent decides that the initialization phase is over (constant fraction of non-$\Collector$--agents is reached) and sets $\Phase{}= 0$ which is then spread via broadcast ($\Phase{}$ is initialized at the beginning of the whole protocol to $-1$).
From there on the $\Clock$-agents use $\Count {}$ to run the leaderless phase clock from \cite{DBLP:conf/soda/AlistarhAG18} for the synchronization which works as follows.
The counter $\Count {}$ is used modulo $\Psi = \Theta(\log n)$.
Whenever two $\Clock$-agents interact, the one with the lower counter value (w.r.t.\ the circular order modulo $\Psi$) increments its $\Count {}$.
If both $\Clock$-agents have the same $\Count {}$ value ties are broken arbitrarily.
When $\Count {}=0$ the variable $\Phase{}$ is increased by one (modulo 10).
Note that, alternatively to this simple clock, any phase clock that requires $\ldauOmicron{\log n}$ states can be used.

\begin{lstalgo}{Clock Synchronization. We assume that $u$ is a $\Clock$ agent.\label{alg:synchronization}}
if $\Phase u = -1$ then$\label{alg:synchronization:1}$
    $\displaystyle \Count u \gets \begin{cases} \Count u + 1& \text{if } \Role v \neq \Collector   \\     \Count u - 1 & \text{if } \Role v = \Collector \text{ and  }\Count u > 0 \end{cases} $
    if $\Count u = 5\cdot \log n$ then
        $\Phase u \gets 0\label{alg:synchronization:4}$

if $\Phase u \neq -1$ and $\Phase v \neq -1$ then
    leaderless_phase_clock ($\Count u, \Count v$) /* execute the leaderless phase clock from \cite{DBLP:conf/soda/AlistarhAG18} */
    if $\Count {u}$ passes through zero then
        $\Phase u \gets \Phase u + 1 \mod 10$
\end{lstalgo}

The $\Tracker$-agents determine which opinion has to take over the role as a challenger (see \cref{alg:tracker}).
The state variable $\Tc{}$ is initialized (see initialization phase) with $1$ and incremented by one (modulo $k$) whenever $\Phase{}$ switches over to zero. Note that during the first tournament $\Tc{} = 2$. This holds due to the initialization of $\Tc{}$ with one and the fact that it is incremented as soon as $\Phase{}$ is incremented from $-1$ to 0 when the initialization phase ends.

\enlargethispage{2\baselineskip}

\begin{lstalgo}{We assume that $u$ is a $\Tracker$-agent.\label{alg:tracker}}
if $\Phase u = 0$ and $u$ interacts fo${}$r the first time in this phase then
    $\Tc u \gets \Tc u + 1$.
\end{lstalgo}

\subsection{Initialization} \label{sec:init}

\label{sec:preprocessing-details}

The objective of this phase is to partition the population into the four different roles $\Collector$, $\Player$, $\Tracker$ and $\Clock$.
Initially every agent has the $\Collector$ role storing one token of its initial opinion.
Whenever two $\Collector$-agents with the same opinion and at most $10$ tokens in total interact, the responder sets its $\Tokens{}$ variable to the sum of the tokens of both agents, and the initiator switches to a roles in $\set{\Clock, \Tracker, \Player}$ uniformly at random.
Agents with opinion $1$ set $\Defender{} = 1$ during their first interactions.
As soon as agent $u$ becomes $\Clock$-agent it uses the state variable $\Count {}$ to determine when the initialization is over by setting $\Phase u$ equals to $0$ which is then  spread via broadcast. 
At this point, the first tournament starts with the $\Setup$ phase.

\begin{lstalgo}{Initialization Phase. We assume that $u$ and $v$ are initially in  $\mathrm{phase}[u] = \mathrm{phase}[v] = -1$. \label{alg:init}}
if $u$ is the initiator fo${}$r the first time and $\Opinion u = 1$ then
    $\Defender u = \ttrue$

if $\Role u = \Role v = \Collector$ and $\Opinion u = \Opinion v$ 
and $\Tokens u + \Tokens v \leq 10$ then
    $(\Tokens u, \Tokens v) \gets (0, \Tokens u + \Tokens v)\label{alg:init:4}$ 
    with probability $1/3$: $\displaystyle \begin{cases} (\Role u, \Count u) & \gets (\Clock, 0) \\ (\Role u, \Tc u) & \gets (\Tracker, 1) \\ (\Role u, \Playeropinion u) & \gets (\Player, U) \end{cases}$

if $\Phase v = 0$ then
    $\Phase u \gets 0$
\end{lstalgo}

\begin{lemma}\label{lem:preprocessing}
  Let $\hat{t}$ denote the interaction, in which the first agent sets $\Phase{}=0$.
  Then, the following statements hold \whp:
  \begin{enumerate}[nosep]
      \item $\hat{t} = \ldauOmicron{n\cdot (k + \log n)}$.
      \item After interaction $\hat{t}$ each of the roles $\Collector,  \Clock, \Tracker$, and $\Player$ are held by at least $n/10$ agents.
      \item After interaction $\hat{t}$ all $\Collector$-agents of opinion $1$ have their defender bit set.
  \end{enumerate}
\end{lemma}
\begin{proof}
We consider a modified process, which mimics the original process. The only difference is that in this process, we prevent $\Clock$-agents from setting their $\Phase{}$ to $0$ by removing \cref{alg:synchronization:4} of \cref{alg:synchronization}.
This causes all agents to remain in the $\Init$ phase, i.e., they have $\Phase{}$ set to $-1$ indefinitely. In this setting all agents keep performing according to \cref{alg:init}. This simplifies the analysis as we do not have to deal with some agents that already started the tournament.
In the following we assume that this modified process runs alongside the original process and that the same random choices are made in both processes. 
Let now $\tau_m(x)$ denote the first interaction in which at most $x \cdot n$ $\Collector$-agents remain in the modified process. Similar, let $\hat{t}_m$ denote the first interaction in which some $\Clock$-agent counts to $5 \log n$. 
Additionally, we define the same notation with subscript $o$ with respect to the original process. 
Observe that $\hat{t}_m = \hat{t}_o$ as until this interaction occurs, both processes are identical.

We start by establishing that, in the modified process, $\tau_m(1/3)$ is reached quickly. That is, the number of remaining $\Collector$-agents decreases fast as long as all nodes follow \cref{alg:init}.

\begin{claim}
\label{claim:lemma1-2}
     It holds that $\tau_m(1/3) = \ldauOmicron{n \cdot k}$ \whp.
\end{claim}
\begin{proof}
In order to reach interaction $\tau_m(1/3)$ exactly $\lceil 2n/3 \rceil$ agents need to leave their $\Collector$ role due to the token transfer in \cref{alg:init:4} of \cref{alg:init}. 
In the following we say that an interaction is \emph{good} if two agents interact that both are $\Collector$-agents, have the same opinion, and have at most $10$ tokens in total. 
Such an interaction decreases the number of $\Collector$-agents by one. 
Let $z_i(t)$ denote the number  $\Collector$-agents of opinion $i$ which have at most $5$ tokens before interaction $t$ is executed. 
If two such agents of the same opinion interact, then the interaction is guaranteed to be good. Fix now some interaction $t < \tau(1/3)$, i.e., an interaction before which more than $n/3$ $\Collector$-agents are still present. Then, the probability for interaction $t$ to be \emph{good} is 
\begin{align*}
    \sum_{i=1}^{k}\frac{z_i(t)}{n} \cdot \frac{z_i(t) - 1}{n - 1} \geq  \frac{1}{n^2} \sum_{i = 1}^{k} z_i(t)^2 - \frac{1}{n^2}\sum_{i=1}^{k} z_i(t) \\
    \overset{(a)}{\geq} \frac{1}{n^2} \frac{\left(\sum_{i=1}^{k} z_i(t)\right)^2}{k} - \frac{1}{n}
    \overset{(b)}{\geq} \frac{1}{n^2} \frac{n^2}{36 \cdot k}  - \frac{1}{n} \overset{(c)}{\geq} \frac{1}{500k}
\end{align*}
For the third inequality (b) we apply the following counting argument to bound $\sum_{i=1}^{k} z_i(t)$: only $n/6$ agents may have at least $6$ tokens as the number of tokens sums to $n$ at all times. 
Observe that we assume that at time $t$ there are still $n/3$ total $\Collector$-agents remaining. 
Hence, $\sum_{i=1}^{k} z_i(t) \geq n/3 - n/6 = n/6$. 
For the last inequality (c), we use that $k\leq n/40$ as assumed in \cref{thm:result-order}.
As each interaction is good with probability at least $1/500k$, independently, we consider a sequence of $500 n k$ interactions and apply Chernoff bounds. 
This yields that, \whp, there will be at least $\lceil 2n/3 \rceil$ good interaction in this sequence, reducing the number of $\Collector$-agents below $n/3$. In other words: $\tau(1/3) < 500 n k$ \whp.
\end{proof}

In the following claim we bound the time for the first $\Clock$-agent to count until $5\log n$ in the modified process.

\begin{claim} 
\label{claim:lemma1}
It holds that $\tau_m(2/3) < \hat{t}_m$ and $\hat{t}_m = \ldauOmicron{n \cdot (k + \log n)}$ \whp
\end{claim}
\begin{proof}
We consider the modified process and couple the counting procedures of any fixed $\Clock$-agent with a biased random walk on the non-negative line. 
The current value of the counter variable corresponds to the position of the walk on the line. 
Each time the $\Clock$-agent interacts as initiator with a non-$\Collector$-agent, the random walk process moves to the right.
Similarly, when interacting with a $\Collector$-agent the random walk moves to the left (or remains at 0 if its current position is 0). 
We are interested in the interactions required for the random walk to hit then value $5 \log n$ as this corresponds to the $\Clock$-agent counting until $5\log n$.
Until $\tau_m(2/3)$ is reached, this hitting time may be minorized with the hitting time of a random walk that has probability exactly $q=2/3$ to move to the left and probability $p=1/3$ to move to the right.
Due to the strong drift towards~0, it is known that such a random walk takes $\poly(n)$ steps \whp to hit $5\log n$.
To determine the constant hidden in $\poly(n)$ we utilize a variant of a known random walk result.
It implies that this hitting time is at least $n^{2.5}$ with probability at least $1-n^{-2.5}$. More details are given in \cref{lem:random-walk} in \cref{apx:auxiliary-results}.
Therefore, \whp, the $\Clock$-agent will \emph{not} reach a counter value of $5\log n$ before, either, $\tau_m(2/3)$ is reached or $n^{2.5}$ interactions have passed. Now, observe that $\tau_m(2/3) < \tau_m(1/3)$ as the number of $\Collector$-agents can only decrease over time. This implies by \cref{claim:lemma1-2} that $\tau_m(2/3) < n^{2.5}$ \whp for large enough $n$. Hence, $\hat{t}_m(2/3)$ precedes $n^{2.5}$ \whp and $\hat{t}_m > \tau_m(2/3)$ follows.

To show the upper bound on $\hat{t}_m$, we first argue that soon after $\tau_m(1/3)$ some $\Clock$-agents increases its counter to $5 \log n$. We follow a similar approach and fix the modified process at some interaction $t \geq \tau_m(1/3)$ together with a $\Clock$-agent and its corresponding random walk. 
This time, we majorize the time for the counter to reach $5\log n$ with the hitting time of a random walk with $p=2/3$ and $q=1/3$. 
Such a random walks is known (e.g Theorem 18.2 of \cite{book/mixingtimes}) to have a hitting time of $\ldauOmicron{\log n}$ \whp. For convenience we included a similar statement in \cref{lem:random-walk} in \cref{apx:auxiliary-results}.
Each movement of the random walk corresponds to one interaction as initiator of the $\Clock$-agent.
As the agent is selected as an initiator with probability $1/n$ in each interaction, it follows from a Chernoff bound that $\ldauOmicron{n \log n}$ interactions guarantee sufficient movements of the random walk \whp.
Therefore, some $\Clock$-agents hits $5 \log n$ before time $\tau_m(1/3) + \ldauOmicron{n\log n}$ \whp. From \cref{claim:lemma1-2} we know that $\tau_m(1/3)  = \ldauOmicron{n \cdot k}$ \whp, allowing us to simplify this upper bound to $\ldauOmicron{n \cdot (k + \log n)}$.
\end{proof}

In order to show the first two statements of the lemma we need the guarantees of \cref{claim:lemma1} in terms of the original process. Initially we established that $\hat{t}_o = \hat{t}_m$ and that until this interaction both processes act identically per definition.
Additionally, note that $\Pr[\tau_o (2/3) = \tau_m(2/3)] \geq \Pr[\tau_m(2/3) \leq \hat{t}_m]$. This inequality holds because, if the event $\tau_m(2/3) \leq \hat{t}_m$ occurs, then both processes acted identically until interaction $\tau_m(2/3)$. Therefore, the amount of $\Collector$-agents is the same in both processes until this interaction, implying that $\tau_m(2/3) = \tau_o(2/3)$. By \cref{claim:lemma1} we have that $\tau_m(2/3) \leq \hat{t}_m$ \whp and therefore $\tau_m(2/3) = \tau_o(2/3)$ is also a high probability event. 
Hence, \whp, \cref{claim:lemma1} also holds when exchanging $\hat{t}_m$ by $\hat{t}_o$ and $\tau_m(2/3)$ by $\tau_o(2/3)$, leading to the statement: $\tau_o(2/3) < \hat{t}_o = \ldauOmicron{n \cdot (k + \log n)}$ \whp.
This inequality immediately yields the first statement of the lemma.
We also use this inequality to show the second statement of the lemma as it implies that at $\hat{t}_o$ at most $2n/3$ $\Collector$-agents remain \whp.
Therefore, at time $\hat{t}_o$, there must be at least $n/3$ $\Collector$-agents that have left their role. Every agent which switches its role selects a new role uniformly and independently at random.
Hence, it follows from Chernoff bounds that each non-$\Collector$ role consists of at least $(n/3) \cdot ( 1 /3) ( 1- o(1)) > n/10$ agents. 
Additionally, note that there must be at least $n/10$ $\Collector$-agents at all times. This follows since there are $n$ tokens in total, and only $\Collector$-agents can hold up to 10 tokens each.

The proof for the final statement of the lemma is straightforward. It suffices to show that every agent interacts at least once before the first $\Clock$-agent sets $\Phase{}$ to $0$. Even if a $\Clock$-agent interacts with a non-$\Collector$-agent each time it is selected as initiator, it takes at least $5 \log n$ such interactions for it to set $\Phase{}$ to $0$. From Chernoff bounds it follows \whp that it requires more than $2  n \log n$ overall interactions for any $\Clock$-agent to be selected as initiator sufficiently many times. However, any fixed agent manages to act as initiator at least once within $2 n \log n$ interactions \whp. As each node is selected with probability $1/n$ as an initiator, the probability that an arbitrary but fixed agent is not selected is at most $(1-1/n)^{2 n \log n} \leq \exp(-2 \log n) \leq n^{-2}$. A union bound over all agents shows that this is enough time for every agent to act as initiator at least once \whp. 
\end{proof}

\begin{lstalgo}[t]{Tournament Algorithm\label{alg:setup}}
if $\Phase u = \Phase v = 0$ then /*Setup Phase*/

    if $\Role u = \Collector$ and $\Role v = \Tracker$ and $\Opinion u = \Tc v$ then 
        $\Challenger u \gets \ttrue$
    
    if $\Role u = \Collector$ then 
        $\displaystyle \ell[u] \gets \begin{cases} \Tokens u & \text{ if } \Defender u \\ -\Tokens u & \text{ if }  \Challenger u \\ 0 & \text{ otherwise. } \end{cases}$

if $\Phase u = \Phase v = 2$ then /* Cancellation Phase */
    
    if $\Role u = \Role v = \Collector$ then
        $\displaystyle (\ell[u],\ell[v]) \gets \left(\floor*{\frac{\ell[u]+\ell[v] }{2}},\ceil*{\frac{\ell[u] +\ell[v]}{2}}\right)$

if $\Phase u = \Phase v = 4$ then /* Lineup Phase */

    if $\Role u = \Collector$ and $\Role v = \Player$ and $\Playeropinion v = U$ then
        $\displaystyle \Playeropinion v \gets \begin{cases} A & \text{ if } \ell[u] > 0 \\ U & \text{ if } \ell[u] = 0 \\ B & \text{ if } \ell[u] < 0 . \end{cases}$
        $\ell[u] \gets \sign(\ell[u]) \cdot (\abs{\ell[u]} - 1)$

if $\Phase u = \Phase v = 6$ then /* Match Phase */

    if $\Role u = \Role v = \Player$ then 
        execute $\Majority(\cS_\Maj)$ /* execute the exact majority protocol from \cite{DBLP:conf/focs/DotyEGSUS21} */

if $\Phase u = \Phase v = 8$ then  /* Conclusion Phase */

    if $\Role u = \Collector$ and $\Role v = \Player$ and $\Playeropinion v = B$ do once$\label{ln:doonceexample:start}$
        $\Defender u \gets \Challenger u$
        $\Challenger u \gets \tfalse$$\label{ln:doonceexample:end}$
        
    if $\Role u = \Collector$ and $\Role v = \Player$ and $\Playeropinion v \in  \{A, U \}$ do once
        $\Challenger u \gets \tfalse$        

if $\Phase v >_{(10)} \Phase u$ then $\label{alg:synchronize1}$
    $\Phase u \gets \Phase v$ $\label{alg:synchronize2}$
\end{lstalgo}

\subsection{Player and Collector Agents}\label{sec:tournament}

The tournaments are performed by both \Player- and \Collector-agents. 
Each tournament is divided into the five phases \Setup, \Cancellation, \Lineup, \Match, and \Conclusion. To synchronize the beginning of the phases we assume that there are phases (numbered with odd numbers) in which non of the  \Player- and \Collector-agents is activated. 

In \Setup $\Collector$-agents determine if their opinion is the challenger (in the $i$-th tournament Opinion $i+1$ is the challenger opinion and $\Tc{}=i$). 
Furthermore, all challenger and defender agents initialize a variable $\ell [u]$ with the (positive or negative) amount of tokens they store.
In \Cancellation the agents use the load balancing protocol from \cite{DBLP:conf/ipps/BerenbrinkFKK19, DBLP:journals/jap/MocquardRSA21}.
At the end of the protocol each agent $u$ will have $\ell[u]\in \set{\overline{\bm{\ell}} -1, \overline{\bm{\ell}}, \overline{\bm{\ell}}+1 }$ where $\overline{\bm{\ell}}$ is the average of all the $\ell [u]$ values from challengers and defender agents rounded to the nearest integer.
This phase is used to reduce the number of tokens such that each token can be assigned to a different \Player-agent.
This will be done in the \Lineup phase. 
The load balancing protocol can be used (see \cite{DBLP:conf/focs/DotyEGSUS21}) to calculate the majority opinion for the case of $k=2$ and large bias.
In that case the majority opinion is the opinion for which a $\Collector{}$-agent exists with $\ell[u]\leq-2$ or $\ell[u]\geq 2$. For the ease of presentation of our protocol we do not distinguish between the case that the majority is already determined after this phase or not. 

In the match phase the $\Player$-agents, now having opinions $A$ (defender opinion), $B$ (challenger opinion) or $U$ (undecided, held by $\Player$-agents which do not receive any opinion) determine the majority opinion using the majority protocol of \cite{DBLP:conf/focs/DotyEGSUS21}. We assume that the protocol returns the result in the state $\Playeropinion{}$ which takes the values of the majority opinion. Note that the protocol from \cite{DBLP:conf/focs/DotyEGSUS21} assumes that each agent has one of the two opinions. In \cref{sec:proof-theorem1} describe in more detail how this protocol can be applied in our setting.
In the \Conclusion phase $\Collector$-agents holding the majority opinion set their defender bit. They have to participate in the next tournament. In \cref{alg:synchronize1,alg:synchronize2} agents broadcast $\Phase{}$ to remain synchronized. 

\subsection{Aftermath}%
\label{sec:proof-sketch}

This subsection provides a short description of how our protocol finishes after the last tournament.

\paragraph{Final Broadcast}
After the final tournament the agents still need to ensure that the ultimate defender -- \whp the initial plurality opinion -- is disseminated to all agents.

The $\Tracker$-agents initiate this final broadcast.
Recall that the $\Tracker$-agents have a variable \Tc{} that keeps track of the challenger in each tournament.
Once this variable reaches $k+1$, all opinions have participated in a tournament, and those $\Collector$-agents that have the defender bit set have \whp the initial plurality opinion.
Now when a $\Tracker$-agent $u$ with $\Tc u = k+1$ interacts with a $\Collector$-agent $v$ with $\Defender{v} = \ttrue$, the defender agent sets its \emph{winner} bit $\Winner v \gets \ttrue$.
This winner bit and the corresponding opinion is disseminated to all agents: any agent $w$ for which $\Winner{w} = \tfalse$ sets $(\Role w, \Opinion w, \Winner w)$ to $(\Collector, \Opinion{v}, \ttrue)$ when it interacts with such a winner agent $v$ (with $\Winner v = \ttrue$).

\paragraph{Proof of \cref{thm:result-order}}
Next we  provide a brief proof sketch for the runtime from the first statement in \cref{thm:result-order} (see \cref{sec:proof-theorem1} for the full proof).
Afterward, we prove the bound on the size of the state space from the first statement in \cref{thm:result-order}.
\begin{proof}[Proof Sketch: Runtime for Statement~(1) of \cref{thm:result-order}]
The proof is done inductively using an invariant
(see \cref{lem:setup} in \cref{sec:proof-theorem1}).
The invariant states that the \Collector{} and \Defender{} bits are set correctly and that the number of $\Player$-agents is sufficiently large for the number of tokens of the (defender and challenger) \Collector-agents.
The rest follows from \cite{DBLP:conf/ipps/BerenbrinkFKK19, DBLP:journals/jap/MocquardRSA21} and \cite{DBLP:conf/focs/DotyEGSUS21}.
\end{proof}

\begin{proof}[Proof: Space Complexity for Statement~(1) of \cref{thm:result-order}]
\Vardef{Shared}
\Cref{fig:state-space} shows a \emph{superset} $\cS$ of our protocol's state space.
Depending on their role, the agents only use a much smaller portion of $\cS$ as described below.

Each agent's state space consists of a set of \emph{shared} variables, which any agent keeps track of, and of \emph{role-specific} variables, which only agents of that role keeps track of.
We use $\cS_{\Shared{}}$ to denote the state set represented by all shared variables and $\cS_{r}$ to denote the variables required for role $r \in \set{ \Clock, \Tracker, \Collector, \Player }$.

Note that $\abs{\cS_{\Shared{}}} = \ldauTheta{1}$.
Indeed, the shared variables encompass the constant size \Role{} variable, the constant size \Phase{} variable, and the constantly many bits required for the do-once statements (see overview of the state space at the beginning of \cref{sec:algorithm}).
The role-specific variables are indicated by the gray boxed in \cref{fig:state-space}.
Specifically:
\begin{itemize}
\item \Clock-agents use \Count{} variable ($\ldauTheta{\log n}$ values).
\item \Tracker-agents use the $\Tc{}$ variable ($k$ values).
\item \Collector-agents use
    the \Opinion{} variable ($k$ values),
    the \Tokens{} variable ($10$ values),
    the \Defender{}, \Challenger{}, \Winner{} bits,
    and the load balancing values $\ell$ ($21$ values).
\item \Player-agents use
    the $\Playeropinion{}$ variable ($3$ values) and
    $O(\log n)$ states for the majority protocol from \cite{DBLP:conf/focs/DotyEGSUS21}.
\end{itemize}
The maximum number of states required by any agent then calculates as

\begin{equation*}
\begin{aligned}
&
\abs{\cS_{\Shared{}}} \cdot \max\set{
    \phantomas[r]{\ldauTheta{\log n}}{\cS_{\Clock}},\;
    \cS_{\Tracker},\;
    \phantomas[r]{k \cdot 10 \cdot 2^2 \cdot 21}{\cS_{\Collector}},\;
    \phantomas[r]{3 \cdot \ldauOmicron{\log n}}{\cS_{\Player}}
}
\\{}={}&
\phantomas[r]{\abs{\cS_{\Shared{}}}}{\ldauTheta{1}}\cdot \max\set{
    \ldauTheta{\log n},\;
    \phantomas[r]{\cS_{\Tracker}}{k},\;
    k \cdot 10 \cdot 2^3 \cdot 21,\;
    3 \cdot \ldauOmicron{\log n}
}
\\{}={}& 
\ldauTheta{k + \log n}
,
\end{aligned}
\end{equation*}

finishing the proof of the first protocol's state complexity.
\end{proof}

	\section{The Improved Algorithm}%
\label{sec:filtering}

The goal in this section is to remove \emph{insignificant} opinions before they even participate in the tournament.
For the moment let us assume that every agent $u$ has a counter $c[u]$ which is used to count the number of interactions with the same opinion.
As soon as the first counter reaches a fixed value $t \in O(\log n)$ the agent sets $\Phase u = 0$ which triggers the beginning of the tournaments. 
Only agents with a counter of at least $t/2$ will participate in the tournament. 
\emph{Insignificant} opinions (those of support $x_i < x_{\max} / c_s$ for some constant $c_s > 1$) are effectively out of the race.
This reduces the amount of required tournaments to $\ldauOmicron{n / x_{\max}}$ and therefore improves the runtime.
To show the correctness of this approach it remains to show that \whp every agent of the initial plurality opinion is among these remaining agents, while no agents of insignificant opinions participate in the tournament.
The rest of the analysis follows along the lines of Statement (2) of \cref{thm:result-order}.
Unfortunately, this simple approach requires an additional counter per agent which exceeds the state space bounds of \cref{thm:result-no-order-phase-clock}.

Our main idea to save on states is to use phase clocks instead of the counters, one per opinion.
In the following we call interactions \emph{meaningful} if an agent interacts with another agent of the same opinion.
We split the agents into \emph{subpopulations}; agents with opinion $i$ belong to subpopulation~$i$.
Every subpopulation runs its own phase clock as follows.
Every agent $u$ has all states of the \emph{junta-driven} phase clock (see \cite{DBLP:journals/dc/AngluinAE08a,DBLP:journals/jacm/GasieniecS21, DBLP:journals/dc/BerenbrinkEFKKR21}), which requires only $\ldauOmicron{\log \log n}$ states compared to the $\Theta(\log n)$ used by the simple counter.
The clocks work as follows.
First, in every subpopulation so-called \emph{junta} agents are selected in meaningful interactions.
Then the phase clock runs on a counter, again in meaningful interactions only.
Note that phase clocks of large subpopulations run faster than phase clocks of small ones.
Whenever a phase clock passes through $0$ the agents increment  $\Phase {}$, which is initialized to $-c$ (we assume that the value $c \in \mathbb{N}$ is a sufficiently large constant). 
Once $\Phase{u}$ becomes $0$ for some agent $u$ this value is broadcasted to all agents as before.
All agents $u$ for which $\Phase{u}$ is still stuck at the initial value $\Phase u = -c$ will not participate in any tournament.
Instead, they change their role (with probability $1/3$ each) from $\Collector$ to $\Clock$, $\Tracker$, or $\Player$.
Note that in contrast to before an agent $u$ does not immediately adopt a new role when it sets $\Tokens u = 0$ in an interaction with another $\Collector$-agent (see \cref{alg:modified-init:7} of \cref{alg:modified-init}).
Instead, agent $u$ waits until $\Phase{u} = 0$.
Then, agent $u$ adopts a new role iff.\ it either has no tokens ($\Tokens u = 0$) or its $\Phase {u}= -c$ (the latter implies that the clock of agent $u$ did not pass through zero even once).

It is now easy to see that this results in a faster convergence time.
Indeed, this follows from how \AlgSimple selects the next challenger opinion if there is no order among the opinions (see description in \cref{sec:removing-the-order}):
In a modified setup phase, a leader selects an opinion as challenger randomly from the \Collector agents which have not yet been defeated in a tournament (using a cascade of one-way epidemic processes on the way).
Hence if there are no \Collector agents left for some of the opinions, there will not be a tournament involving that opinion, and thus the total runtime will be reduced accordingly.

As soon as the first agent reaches $\Phase{u} = 0$ all agents proceed with the modified version of \AlgSimple.
We remark that it can happen that only $o(n)$ $\Collector$-agents remain after removing all insignificant opinions.
In this case, the $\Cancellation$ phase will not achieve a balanced state.
However, all tokens will fit into the \Player agents nonetheless, as we will show in Statement (3) of \cref{lem:init-filtering} that there will be a constant fraction of agents for each role in $\set{\Clock, \Tracker, \Player}$.

While the overall approach sounds very easy, the crux lies in the analysis. 
First of all, we have to analyze the \emph{speed} of the clocks running via meaningful interactions only (\cref{lem:clock-subpopulation})
Then we have to show that all agents of the plurality opinion pass through $0$ at least once, meaning they will participate in the tournament (\cref{lem:init-filtering}).
Finally, we have to show that all agents with insignificant opinions will not participate in any tournament, either because they did not finish the \FormJunta protocol (\cref{lem:small-population}) or because their phase clock runs too slow (\cref{lem:init-filtering}).

\paragraph{Junta-Driven Phase Clock}%
\label{paragraph:clock}
We use the phase clock implementation from \cite{DBLP:journals/dc/BerenbrinkEFKKR21} which starts by electing a junta.
We select the junta in exactly the same way but using meaningful interactions only.
Each agent is equipped with a $\Level{}$ variable, which is initially $0$, and a bit which indicates whether the agent is still active.
Agents progress through levels:
They are initially active, and they remain active and increase their level as long as they interact (as initiators) with another agent on the same or on a higher level. 
If they initiate an interaction with another agent on a lower level, they become inactive.
Finally, agents become also inactive if they hit the maximum level $\ell_{\max} = \floor{\log \log n} - 3$.
All agents that reach this maximum level form the junta and start the phase clock protocol.

In the phase clock every agent is equipped with a phase counter $p[u]$ (initially $0$).
Whenever a junta agent $u$ initiates an interaction with an agent $v$ it sets $p[u] = \max\{p[u], p[v] + 1\}$.
If the initiating agent $u$ is not a junta agent, then $u$ sets $p[u] = \max\{p[u], p[v]\}$.
For $i>0$, we say that an agent $u$ \emph{passes through zero} for the $i$-th time if its phase counter $p[u]$ fulfills $\floor{ p[u] / m } \geq i$ for the first time ($m \in \mathbb{N}$ is a fitting large enough constant). Note that in \cite{DBLP:journals/dc/BerenbrinkEFKKR21} the same property is referred to as $u$ reaching hour $i$ for the first time.

In our protocol we set the maximum level to $\ell_{\max} = \floor{ \log \log n }  - 2$. We show in the proof of \cref{lem:clock-subpopulation} that this modified maximum level still allows the election of a junta \whp as long as the subpopulation has size at least $\sqrt{n}$.

We denote by $\cS_c$ the $\Theta(\log \log n)$ states that are required to execute the junta election and phase clock protocols.
We assume that all agents are initially equipped with sufficiently many additional states to run this clock.
As soon as an agent $u$ sets $\Phase{u}$ to $0$ it may reuse these states.
The following lemma states properties of this phase clock.

\begin{lstalgo}[t]{Modified Initialization.  We assume that $\mathrm{phase}[u] < 0$. \label{alg:modified-init}}
if $\Opinion u = \Opinion v$  and $\Phase v <0$ then
    form_junta_protocol ($\cS_{c}$)/* execute the junta-election protocol from \cite{DBLP:journals/dc/BerenbrinkEFKKR21} */ 
    loglog_phase_clock ($\cS_{c}$) /* execute the phase clock protocol from \cite{DBLP:journals/dc/BerenbrinkEFKKR21} */

    if phase clock of $u$ passes through zero then
        $\Phase u \gets \Phase u + 1$

    if $\Tokens u + \Tokens v \leq 10$ then
         $(\Tokens u, \Tokens v) \gets (0, \Tokens u + \Tokens v)\label{alg:modified-init:7}$

if $\Phase u = 0$ or $\Phase v = 0$ then
    if $\Phase u = -c$ or $\Tokens u = 0\label{alg:modified-init:9}$
        with probability $1/3$: $\displaystyle \begin{cases} (\Role u, \Count u) & \gets (\Clock, 0) \\ (\Role u, \Tc u) & \gets (\Tracker, 1) \\ (\Role u, \Playeropinion u) & \gets (\Player, U) \end{cases}\label{alg:modified-init:10}$
    $\Phase u \gets 0$
\end{lstalgo}

\begin{lemma}
\label{lem:clock}
Assume that we run the junta-election process and phase clock from \cite{DBLP:journals/dc/BerenbrinkEFKKR21} on a population of $n$ agents.
Let $s(0)$ ($e(0)$, resp.) be the interaction when the first (last, resp.) junta agent is elected and
let $s(i)$ ($e(i)$, resp.) be the interaction when the first (last, resp.) agent passes through zero for the $i$-th time.
Then, for any constant $a>0$, there exist two properly chosen constants $c_1'$ and $c_2'$, such that we have with probability at least $1-n^{-a}$, 
\begin{enumerate}[nosep]
\item The protocol elects a non-empty junta of size at most $n^{0.98}$.
\item $s(0) \leq c_2' \cdot n \log (n)$.
\item $c_1' \cdot n \log n \leq s(i+1)-s(i) \leq c_2' \cdot n \log n $ \quad for any $i = O(\poly(n))$,
\item $s(i+1) > e(i)$ \quad for any $i = O(\poly(n))$.
\end{enumerate}
\end{lemma}
\begin{proof}
Follows from Theorem 1 and Lemma 6 in \cite{DBLP:journals/dc/BerenbrinkEFKKR21}.
\end{proof}

We denote by $s_j(0)$ ($e_j(0)$) be the interaction at which the first (last, respectively) junta agent is elected in subpopulation $j$.
For $i>0$ we denote by $s_j(i)$ ($e_j(i)$) the time when the first (last, respectively) agent of opinion $j$ passes through zero for the $i$-th time.
The following lemma adjusts the results of \cref{lem:clock} to subpopulations. 

\begin{lemma}
\label{lem:clock-subpopulation}
Fix a subpopulation $j$ and assume that  $x_j \geq n^{1/2}$. Consider the phase clock driven by subpopulation $j$. Then, for any constant $a>0$, there exist constants $c_1\le c_2 \in \mathbb{N}$ such that the following statements hold with probability $1-{x_j}^{-a}$.
\begin{enumerate}[nosep]
\item Subpopulation $j$ elects a non-empty junta  with at most $(x_j)^{0.98}$ agents.
\item $s_j(0) \leq c_2 \cdot \frac{n^2}{x_j} \log (n)$.
\item  $ c_1\cdot\frac{n^2}{x_j} \log(n)\le s_j(i+1) -  s_j(i) \le c_2\cdot  \frac{n^2}{x_j} \log(n)$ \quad for any $i = O(\poly(n))$.
\item $s_j(i+1) > e_j(i)$ \quad for any $i = O(\poly(n))$.
\end{enumerate}
\end{lemma}
\begin{proof}
The statements of this lemma would directly follow from  \cref{lem:clock} by replacing $n$ with $x_j$. 
However, the junta-election mentioned in \cref{lem:clock} assumes that a maximum level $\lfloor \log \log x_j \rfloor - 3$ is set. 
As our agents do not know the value $x_j$, we set this level to $\ell_{\max} = \lfloor \log \log n \rfloor -2$ instead.
With the following claim we show that this modification still leads to a junta of desired size if $x_j \geq \sqrt{n}$ 

\begin{claim}    
\label{claim:claim-without-name}
If $x_j \geq \sqrt{n}$ then the \FormJunta protocol \cite{DBLP:journals/dc/BerenbrinkEFKKR21} configured with maximum level $\ell_{\max} = \lfloor \log \log n \rfloor - 2$ elects a non-empty junta of $\leq x_j^{0.98}$ agents within $\ldauOmicron{x_j \log(x_j)}$ meaningful interactions and with probability at least $1-x_j^{-a}$ (for any constant $a > 0$).
\end{claim}
\begin{proof}
We start by showing the bounds on the junta size. Depending on the size $x_j \geq \sqrt{n}$ of the subpopulation $j$, we can express $\ell_{\max}$ as either (i) $\lfloor \log \log x_j \rfloor - 3$, or (ii) $\lfloor \log \log x_j \rfloor - 2$. Consider the first case. In this case, $\ell_{\max}$ matches the maximum level in specification of the \FormJunta \cite{DBLP:journals/dc/BerenbrinkEFKKR21} protocol for populations of size $x_j$. Therefore, we can apply the corresponding Theorem 1, which states that a non-empty junta of size $\leq x_j^{0.98}$ is formed with probability $1-x_j^{-a}$ (for any constant $a>0$).

In the other case $\ell_{\max} = \lfloor \log \log x_j \rfloor -2$. 
Throughout the \FormJunta process, only active agents may modify their level. 
That is, if an active agent $u$ initiates a meaningful interaction with a node $v$, then (i) it becomes inactive if $v$ has a level lower than $u$, or (ii) it remains active otherwise.\footnote{Note that the state transitions for agents on the first level $0$ are slightly different but not relevant for this proof.}
Just as in \cite{DBLP:journals/dc/BerenbrinkEFKKR21}, we denote by $B_\ell$ the number of agents that reach at least level $i$. Per definition, it must hold that $B_\ell \geq B_{\ell + 1}$ for any level $\ell \geq 0$.
First, we show that between $1$ and $x_j^{0.98}$ agents make it to level $\ell_{\max}$ with probability $1 - x_j^{-a}$ (for any constant $a>0$). 
The upper-bound on this number follows directly from Lemma 5 of \cite{DBLP:journals/dc/BerenbrinkEFKKR21}. It states that $B_{\lfloor \log \log x_j \rfloor - 3} < x_j^{0.98}$ with probability $1 - x_j^{-a}$ (again for arbitrary constants $a > 0$).
Due to the monotonicity of $B_\ell$, it follows that $B_{\ell_{\max}} < x_j^{0.98}$ as well. In order to show the lower-bound $B_{\ell_{\max}} > 1$ we would like to use Lemma 4 of \cite{DBLP:journals/dc/BerenbrinkEFKKR21}. Unfortunately, it only yields that $B_{\ell_{\max}-1} > 1$. Fortunately, in the proof of Lemma 4 they show the slightly stronger statement of $B_{\ell_{\max}-1} > x_j^{2/3}$ with probability at least $1-x_j^{-a}$.
We argue that this implies that $B_{\ell_{\max}}>1$ with probability $1-x_j^{-a}$. To show this, we rely on the coupling idea described in Footnote 6 on page 100 of \cite{DBLP:journals/dc/BerenbrinkEFKKR21}. 
That is, we serialize the points in time $\{t(l)\}_{l=1}^{x_j^{2/3}}$ at which the first $x_j^{2/3}$ agents that entered level $\ell_{\max}-1$ make their first interaction as an initiator.
At time $t(l)$, the $l$-th such agent decides whether it stays active and progresses to level $\ell_{\max}$ or becomes inactive (according to (i) and (ii) above). 
Observe that for any such agent that makes its decision after $t(x_j^{2/3} / 2)$, the probability to remain active is at least $x_j^{2/3} / (2x_j)$ (as at this time already $x_j^{2/3} / 2$ agents entered level $\ell_{\max}-1$). 
Hence, in expectation, at least $x_j^{2/3} / (2x_j) \cdot x_j^{2/3} = x_j^{1/3} / 2$ agents progress to level $\ell_{\max}$. 
From Chernoff bounds it follows that at least $x_j^{1/3}(1-o(1)) / 2$ agents reach $\ell_{\max}$ with probability $1 - x_j^{-\omega(1)}$.

It remains to show that $\ldauOmicron{x_j \log (x_j)}$ meaningful interactions suffice for the first agent to reach level $\ell_{\max}$.
This follows from Lemma 3 of \cite{DBLP:journals/dc/BerenbrinkEFKKR21}. There it is shown that even if the maximum level is unbounded, all nodes become inactive within $\ldauOmicron{x_j \log (x_j)}$ interactions and with probability at least $1-x_i^{-a}$. 
\end{proof}

Statement (1) now follows directly from this claim. For Statement (2) we also refer to this claim and note that the junta-election is driven in the subpopulation. 
Hence, the $\ldauOmicron{x_j \log (x_j)}$ \emph{meaningful interactions} need to be converted into global interactions. 
To that end, observe that $(n^2 / x_j) \cdot (1 + o(1))$ global interactions suffice for $x_j$ meaningful interactions to occur with probability $1-x_j^{-\omega(1)}$.
Because the probability for any fixed interaction to be meaningful is $x_j^2 / n^2$, this immediately follows from Chernoff bounds. 
A symmetric approach also yields that at least $(n^2/ x_j) \cdot (1-o(1))$ global interactions are required for $x_j$ many meaningful interactions to occur.
This implies that $s_j(0) = \ldauOmicron{(n^2 / x_j) \cdot  \log x_j}$. Due to the constraint on $x_j$, it holds that $\log(n) \geq \log(x_j) \geq \log(n) / 2$ and Statement (2) follows.

The proof of Statement (3) follows from Statement (3) of \cref{lem:clock} and a   conversion to global interactions.
Additionally, observe that due to the constraint on $x_j$ we have $\poly(n) = \poly(x_j)$ and note that the constant hidden in the exponent of $\poly(n)$ in \cref{lem:clock} can be made arbitrary large. 
The proof of Statement (4) again directly follows from \cref{lem:clock} together with above observation of $\poly(n) = \poly(x_j)$.
\end{proof}

\begin{lemma}
\label{lem:small-population}
Fix a subpopulation $j$ of $x_j \leq \sqrt{n}$ agents. Let $\varepsilon >0$ be an arbitrary small constant. Then, subpopulation $j$ will not elect a  junta agent before interaction $n^{1.5 - \varepsilon}$ with probability $1-n^{-\omega(1)}$.
\end{lemma}
\begin{proof}
    In order to join the junta, agents need to increase their level from $0$ to $\ell_{\max} = \lfloor \log \log n \rfloor -2$. 
    Per definition of the junta election \cite{DBLP:journals/dc/BerenbrinkEFKKR21}, an agent $u$ may only increase its level if it interacts as an initiator (and some additional conditions hold).
    Furthermore, this increase is at most an increment of $1$. 
    Therefore, any fixed agent $u$ of subpopulation $j$ requires at least $\ell_{\max}$ meaningful interactions as an initiator to join the junta. 
    In the following we call such an interaction \emph{bad}.
    The probability that any fixed interaction is bad is $(1/n) \cdot x_j / n \le n^{-1.5}$. 
    Let $\varepsilon > 0$ be an arbitrary small constant. We show that in a sequence of $n^{1.5 - \varepsilon}$ there will be less than $\ell_{\max}$ bad interactions with probability $1-n^{-\omega(1)}$. 
    The lemma's statement then follows from a union bound over all agents in subpopulation $i$.
    
    The number of bad interactions of $u$ in this may be majorized by $\text{Bin}(n^{1.5-\varepsilon}, n^{-1.5})$. It holds that
    \begin{align*}
        &\Pr\left[\text{Bin}(n^{1.5-\varepsilon}, n^{-1.5}) \geq \ell_{\max} \right] \\&= \sum_{i=0}^{n^{1.5 - \varepsilon} - \ell_{\max}} \Pr\left[\text{Bin}(n^{1.5-\varepsilon}, n^{-1.5}) = \ell_{\max} + i \right] \\& \overset{(a)}{\leq} n^{1.5 - \varepsilon} \cdot \Pr\left[\text{Bin}(n^{1.5-\varepsilon}, n^{-1.5}) = \ell_{\max} \right]
    \end{align*}
    In step (a) we use that $\ell_{\max}$ is much larger than the expected value of this distribution.
    Hence, the terms in the sum decline with further $i$. This allows us to upper-bound each term in the sum by $p=\Pr \left[ \text{Bin}(n^{1.5-\varepsilon}, n^{-1.5}) = \ell_{\max} \right]$. 
    Using the PDF of the binomial distribution we can further bound $p$.
    \begin{align*}
        p &=\binom{n^{1.5 - \varepsilon}}{\ell_{\max}} \cdot (n^{-1.5})^{\ell_{\max}} \cdot (1-n^{-1.5})^{n^{1.5 - \varepsilon} - \ell_{\max}} \\&\leq \left( \frac{e \cdot n^{1.5 - \varepsilon}}{\ell_{\max}} \right)^{\ell_{\max}} (n^{-1.5})^{\ell_{\max}} 
        = \left(\frac{e}{\ell_{\max}}\right)^{\ell_{\max}} \cdot \frac{1}{n^{\varepsilon \cdot \ell_{\max}}}
    \end{align*}
    Since $\ell_{\max} = \Theta(\log \log n)$, this implies that $p = n^{-\omega(1)}$. Hence, $\Pr \left[\text{Bin}(n^{1.5-\varepsilon}, n^{-1.5}) \geq \ell_{\max}  \right] = n^{1.5 - \varepsilon} \cdot n^{-\omega(1)} = n^{-\omega(1)}$ for sufficiently large $n$ and the result follows.
\end{proof}

In the following we define $T_i(t)$ as the total number of tokens for opinion $i$ at interaction $t$, i.e., \[T_i(t) := \sum_{\set{u | \Opinion{u}(t) = i}}\Tokens{u}(t)\] where $\Opinion{u}(t)$ and $\Tokens{u}(t)$ denote the values of the variables $\Opinion{u}$ and $\Tokens{u}$, respectively, in interaction $t$.
Note that $T_i(0)$ is the initial support of opinion $i$.
\begin{lemma} 
\label{lem:init-filtering}
    Assume that $x_{\max} > n^{1/2 + \varepsilon}$ for a small constant $\varepsilon > 0$. Let $i$ be the initial plurality opinion and let $\hat{t}$ denote the first interaction in which $\Phase{}= 0$ for all agents. Then, \whp, $\hat{t} = \Theta( (n^2 / x_{\max}) \cdot \log n)$ and the following holds
    after interaction~$\hat{t}$ \whp:
    \begin{enumerate}[nosep]
        \item There are at most $\ldauOmicron{n / x_{\max}}$ distinct opinions left.
        \item For the initial plurality opinion $i$ it holds that $T_i(\hat t) = T_i(0)$. 
        \item Each of the roles $\Clock, \Tracker$, and $\Player$ is held by at least $n/10$ agents.
    \end{enumerate}
\end{lemma}

\begin{proof}
We first show the bound on $\hat t$.
Recall that $s_i(0)$ is defined as the interaction when the first junta agent in subpopulation $i$ is elected, and $s_i(c)$ is defined as the interaction when the clock of the first agent of opinion $i$ ticks for the $c$-th time.
We will prove upper and lower bounds for $\hat{t}$ based on $s_i(c)$.

From Statements (1) and (3) of \cref{lem:clock-subpopulation} (with $a=4$) it follows that $s_i(c) \leq  (c+1) c_2 \cdot\frac{n^2}{x_{\max}}\log n$ with probability at least $1-(1+c)\cdot x_{\max}^{-4}\ge 1-( 1+c)\cdot n^{-2-4\varepsilon}$ (since we assume that $x_i \geq n^{1/2 + \varepsilon}$).
Once an agent $u$ has reached $\Phase u = 0$, this phase value is disseminated to all other agents via one-way epidemics.
It follows that, \whp, $\hat t \leq s_i(c) + \tau_{BC}$, where $\tau_{BC}$ is the \emph{broadcast time} with $\tau_{BC} \leq c_2 \cdot n^2/x_{\max} \log n$ \whp \cite{DBLP:journals/dc/AngluinAE08a}.
Ultimately, $\hat t \leq c_2 \cdot (c+2) \cdot n^2/x_{\max}\cdot\log n$ \whp.

For the lower bound, we  observe that $\hat{t}\ge s_j(c)\ge c \cdot c_1 \cdot n^2/x_{j}\cdot\log n \geq c \cdot c_1 \cdot  n^2/x_{\max}\cdot\log n$ with probability at least $1 - c\cdot n^{-2-4\varepsilon}$.
A union bound over all opinions yields $\hat t \geq c \cdot c_1 \cdot n^2/x_{\max}\cdot\log n$ \whp. Together with the upper bound, the result for $\hat t$ follows. Next, we show the three statements individually.

\paragraph{Statement~(1)}
Let $c^{*}= (c+2)\cdot c_2$ be the constant from the upper bound on $\hat{t}$ and define $c_s = c^{*} / c_1$.
In the following, we show that any insignificant opinion $j$ vanishes.
For this, let $j$ be an arbitrary but fixed opinion with $x_j < x_{\max} / c_s$.
We distinguish two cases.

\paragraph{Case 1: $x_j \geq \sqrt{n}$}
From \cref{lem:clock-subpopulation} we get that \whp
\begin{align*}
\phantomas[r]{s_j(c)}{\hat{t}}
\leq
\phantomas[l]{c \cdot c_1 \cdot \frac{n^2}{x_{\max}} \log n}{c^* \cdot \frac{n^2}{x_{\max}} \log n}
&, &
s_j(1) \geq
c_1 \cdot \frac{n^2}{x_j} \log n > c^* \cdot \frac{n^2}{x_{\max}} \log n,
\end{align*}
where the last inequality uses the definition of $c_s$ and $x_j < x_{\max} / c_s$.
Together with a union bound this implies that \whp the clocks of all agents of opinion $j$ do not tick even once.
Hence, opinion $j$ vanishes at latest in interaction $\hat t$ \whp.

\paragraph{Case 2: $x_j < \sqrt{n}$}
Similarly to before, we get from above bounds on $\hat t$ and from \cref{lem:small-population} that \whp
\begin{align*}
\phantomas[r]{s_j(c)}{\hat{t}}
=
\phantomas[l]{c \cdot c_1 \cdot \frac{n^2}{x_{\max}} \log n}{\ldauOmicron{n^{3/2 - \varepsilon} \cdot \log n}}
&&\text{and} &&
s_j(1) \geq
\phantomas[l]{c_1 \cdot \frac{n^2}{x_j} \log n > c^* \cdot \frac{n^2}{x_{\max}} \log n,}{s_j(0) \geq n^{3/2 - \varepsilon'}.}
\end{align*}
Together with $\varepsilon'<\varepsilon$ and a union bound this again implies that \whp the clocks of all agents of opinion $j$ do not tick even once.
Hence, also in this case opinion $j$ vanishes at latest in interaction $\hat t$ \whp.

\smallskip

Together the two cases show that any opinion $j$ with $x_j<x_{\max}/c_s$  \whp does not compete in the tournaments.
Since we have $n$ agents, at most $n \cdot c_s/x_{\max}=\ldauOmicron{n/x_{\max}}$ opinions remain after $\hat{t}$ interactions \whp.

\paragraph{Statement~(2)}
To show the statement we need to show that the clocks of any agent of the initial plurality opinion $i$ pass through zero at least once before the first agent $u$ hits $\Phase{u} = 0$.
Recall that $s_j(c)$ is the interaction when the clock of any agent with opinion $j$ passes through zero for the $c$-th time (this is the first interaction when any agent $u$ sets $\Phase u = 0$).

In the following we only consider opinions $j$ with $x_j = \Omega(x_i) $. The statement for smaller opinions follows from the above proof of Statement~(1).
There we have shown that the clocks of agents of smaller opinions will not pass through zero even once before interaction $\hat t$.
For significant opinions $j$, we observe $s_j(c)\geq c\cdot c_1 \cdot n^2/x_{\max}$ \whp as shown in the beginning of the proof.
From the bound on $s_j(c)$ and from \cref{lem:clock-subpopulation} we get that \whp
\begin{align*}
s_j(c) \geq c \cdot c_1 \cdot \frac{n^2}{x_{\max}} \log n
&&\text{and} &&
e_i(1) \leq
\phantomas[l]{c_1 \cdot \frac{n^2}{x_j} \log n > c^* \cdot \frac{n^2}{x_{\max}} \log n,}{s_i(2) \leq 3c_2 \cdot \frac{n^2}{x_{\max}} \log n.}
\end{align*}
By choosing a sufficiently large constant $c > 3c_2/c_1$ in \cref{alg:modified-init}, this yields $s_j(c) > s_i(2)$ \whp.
In other words, at the time when the first agent's clock has passed through zero for the $c$-th time, the clocks of all agents of opinion $i$ have passed through zero at least once.
In particular, $\Phase u \neq -c$ for any agent $u$ with opinion $i$ in that interaction.

The total number of tokens $T_i(t)$ of opinion $i$ can only change in some interaction $t$ if an agent $u$ of opinion $i$ adopts another role in \cref{alg:modified-init:10} of \cref{alg:modified-init} while $\Tokens u > 0$.
However, we have just shown that when the first agent $v$ sets $\Phase v = 0$, any agent $u$ of opinion $i$ has $\Phase{u} \neq -c$.
Hence it follows that agent $u$ can adopt a different role in 
\cref{alg:modified-init:10} of \cref{alg:modified-init} only if agent $u$ had $\Tokens u = 0$ in \cref{alg:modified-init:9} of \cref{alg:modified-init}.
Therefore, such an interaction does not change the total number of tokens for opinion $i$ and the statement follows.

\paragraph{Statement~(3)}
The proof follows from similar arguments as the proof of Statement (2) of \cref{lem:preprocessing}.
\end{proof}

\subsection[Proof]{Proof of \cref{thm:result-no-order-phase-clock}.}

We split the proof of \cref{thm:result-no-order-phase-clock} into three parts, the proof of the correctness of the result, the proof of the runtime, and the proof of the state space requirements.
Essentially, the theorem follows from \cref{lem:init-filtering} for the correctness of the modified initialization phase (\cref{alg:modified-init}) and from Statement (2) of \cref{thm:result-order} for the correctness of \AlgSimple.
\begin{proof}[Proof of \cref{thm:result-no-order-phase-clock}, Correctness of the Result]
\mbox{} \\
In \AlgComplex, all agents start with the modified initialization phase defined in \cref{alg:modified-init}.
After this initialization, they execute the tournament according to the variant of \AlgSimple which does not need an order among the opinions (see \cref{sec:removing-the-order}).
By Statement (3) of \cref{lem:init-filtering} we get that all roles in $\set{\Clock,\Tracker,\Player}$ are held by at least a constant fraction of agents at time $\hat t$.
The number of agents with role $\Collector$ may be asymptotically much smaller, however, their number does not affect the outcome of \AlgSimple.
Statement (2) of \cref{lem:init-filtering} guarantees that at the beginning of the tournaments the initial plurality still has all of its initial tokens.
It follows along the lines of the proof of Statement (1) of \cref{thm:result-order} that this opinion will be the defender at the end of the tournament, and all agents will output this opinion after the final broadcast as described in \cref{sec:proof-sketch}.
\end{proof}

\vspace{-1\baselineskip}

\begin{proof}[Proof of \cref{thm:result-no-order-phase-clock}, Runtime of the Algorithm]
\mbox{} \\
From \cref{lem:init-filtering} we get that, \whp, after $\hat{t} = \ldauOmicron{ n^2 / x_{\max}  \cdot \log n}$ interactions all agents $u$ have $\Phase u = 0$ in \cref{alg:modified-init}.
The protocol then proceeds according to the variant of \AlgSimple  which does not require an order among the opinions described in \cref{sec:removing-the-order}.
By Statement (2) of \cref{lem:init-filtering}, at most $\ldauOmicron{n / x_{\max}}$ opinions have at least one \Collector agent each, \whp.
If there is not a single \Collector agent left for some opinion, this opinion cannot become a challenger in any of the tournaments.
Therefore, the total number of tournaments executed in \AlgSimple is bounded \whp by $\ldauOmicron{n / x_{\max}}$.
As argued in the proof of Statement (1) of \cref{thm:result-order}, each tournament takes $\ldauOmicron{n \log n}$ interactions \whp, and the modified \AlgSimple also needs to perform a leader-election, which takes $\ldauOmicron{n \log^2 n}$ interactions \cite{DBLP:journals/jacm/GasieniecS21}.
Together, we conclude that \AlgComplex has a runtime of $\ldauOmicron{ n^2 / x_{\max} \cdot \log n + n \log^2 n }$ interactions \whp.
\end{proof}

\vspace{-1\baselineskip}

\begin{proof}[Proof of \cref{thm:result-no-order-phase-clock}, States of the Algorithm]
\mbox{} \\
\AlgComplex requires the states used in the modified initialization (\cref{alg:modified-init}) and the states used by \AlgSimple (Statement (2) of \cref{thm:result-order}).
In \cref{alg:modified-init}, all \Collector agents need to store the set of states $\cS_{c}$ of size $\Theta(\log \log n)$ required to run the junta-based phase clocks. Additionally, the size of the $\Phase{}$ variable is increased by a constant, starting now at $-c$.
The remaining states have the same size as in \AlgSimple.
Together, this gives us the claimed state space size of $\Theta( k \cdot \log \log n + \log n)$.
\end{proof}

\section{Conclusions}

We present population protocols that efficiently solve exact plurality consensus with high probability.
While it is known that always correct, exact plurality with $k$ opinions needs $\Omega(k^2)$ states per agent, we show that a small failure probability leads to efficient exact plurality consensus with $\ldauOmicron{k + \log n}$ states, \whp.
The runtime can further be reduced at the cost of a small additional factor of $\ldauOmicron{\log \log n}$ states.

Our protocols use majority, leader election, and junta election protocols as a black box.
Improving the guarantees of these black boxes would also carry over to our protocols.
For example, a leader election protocol that has a \emph{with high probability} runtime of $\ldauOmicron{\log n}$ would immediately improve our runtime.
Similarly, both a constant state majority protocol and a constant state junta election protocol (that work \emph{with high probability}) would immediately improve our state space bounds.
Furthermore, we believe that
$\Omega(n/x_{\max})$ is a natural lower bound for the runtime, and thus the possible improvements mentioned above would lead to a state- and time-optimal exact plurality consensus protocol.
\enlargethispage{4\baselineskip}

In our main result we prune small opinions in order to reduce the number of tournaments.
We conjecture that this yields almost optimal protocols.
In order to further improve the runtime (possibly at the expense of a slightly increased state complexity) we believe that additional techniques are required.
In particular, it would be interesting to find another, more efficient way than pairwise comparison of opinions via tournaments to identify the plurality opinion.

	\printbibliography
\newpage
\appendix
\section*{Appendix}

\section{Technical Details of the \AlgSimple Analysis}
\label{sec:proof-theorem1}

\paragraph{Application of the Majority Protocol}
\label{obs:apply-majority-protocol}
In this paragraph, we present how the exact majority protocol from~\cite{DBLP:conf/focs/DotyEGSUS21} is integrated into our \AlgSimple.
That protocol determines the majority between $n$ agents having either opinions $A$ or $B$.
Every agent $u$ has a variable called $\Output{u}$ which finally stores (unless there is a tie) the majority opinion.
The initial opinion is stored in $\Input{u}$. 
The algorithm uses a variable $\Bias{u}$ and sets $\Bias{u} = +1$ if $\Input{u}=A$  and $\Bias{u} = -1$ if $\Input{u} = B$.
In our protocol we execute the exact majority protocol among the $\Player$-agents only. Hence, each $\Player$-agent needs the same set of states (additionally to the ones given in \cref{sec:algorithm}) as the exact majority protocol from \cite{DBLP:conf/focs/DotyEGSUS21}.  
\AlgSimple now is initialized as follows. 
A $\Player$-agent $u$ with $\Playeropinion{u} \neq U$ sets $\Input{u} = \Playeropinion{u}$.
A $\Player$-agent $u$ with $\Playeropinion{u} = U$ sets $\Bias{u} = 0$.
With this initialization the protocol determines the majority in time $\ldauOmicron{n\log n}$, since the number of $\Player$-agents is at least $\ldauOmega{n}$.
Note that in contrast to \cite{DBLP:conf/focs/DotyEGSUS21} we do not need the slow and always correct algorithm used since we are only interested in results that hold with high probability.
We assume for every $\Player$-agent $u$ that $\Playeropinion{u}$ stores the output of the protocol.

\paragraph{Proof of \cref{thm:result-order}, Runtime of Statement (1)}

We prove \cref{thm:result-order} via an induction using the following lemma which provides an invariant for our algorithm.
In the following we assume that the phase clocks are synchronized, the length of the phases is sufficient, and   that the even phases are separated from each other. This follows from \cite{DBLP:conf/soda/AlistarhAG18}. 
For $1\le j<k, 0\leq i \leq 9$ let $t_i(j)$ be
the interaction in which the first agent enters phase $i$ for the $j$-th time. Let $\ell_j$  
be the plurality opinion $\ell$ among $1,\ldots j$.

\begin{lemma}
\label{lem:setup}
 Fix a $j$ with $1\le j<k$ and assume that the first $j-1$ tournaments worked correctly. 
 Then we have \whp

\begin{enumerate}\itemsep0pt

    \item At time $t_2(j)$ all $\Collector$-agents $u$ with  opinion  $j+1$ have $\Challenger{u}=\ttrue$. All other $\Collector$-agents have $\Challenger{u}=\tfalse$.
    Furthermore, all $\Collector$-agents $v$ not having opinion $\ell_{j}$ have $\Defender{v}=\tfalse$.
    \item Let $\mathcal{A}$ be the set of agents $u$ with $\Playeropinion{u}=A$ and let $\mathcal{B}$ be the set of agents with $\Playeropinion{u}=B$. 
    At time  $t_6(j)$ we have $|\mathcal{A}| \geq |\mathcal{B}|$ iff $x_{\ell_j}(0) \geq x_{j+1}(0)$.

    \item If $|\mathcal{A}|\ge|\mathcal{B}|$ ($|\mathcal{A}|<|\mathcal{B}|$) at time $t_8(j)$ we have  $\Playeropinion{u} \in \{A,U\}$ ($\Playeropinion{u}=B$) for all $\Player$-agents $u$.
    \item At time $t_0(j+1)$ all $\Collector$-agents $u$ with  opinion  $\ell_{j+1}$ have $\Defender{u}=\ttrue$.
\end{enumerate}

\end{lemma}
\begin{proof}
From \cref{lem:preprocessing} it follows that each role \Collector, \Player, \Clock, and \Tracker is held by at least $n/10$ agents, \whp.
We denote the set of $\Player$-agents by $P$ and the set of $\Collector$-agents by $C$.
In the following we prove the statements one after the other.

\paragraph{Statement~(1)}
 First we show that in Phase~0 of tournament $j$ each agent interacts at least twice with a $\Tracker$-agent.  \whp we have at least $n/10$ $\Tracker$-agents, hence the probability to interact in a fixed step with a $\Tracker$-agent is at least $1/10$. The claim now follows from Chernoff bounds.
 
 Fix an agent $u$ with opinion $j+1$. Since agent $u$ interacts at least once with a $\Tracker$-agent in Phase~0, $u$ sets $\Challenger{u}=\ttrue$ in Line $3$ of \cref{alg:setup}.
 $\Defender{u} =\tfalse$ follows from the initialization phase (see \cref{alg:init}).
 Now consider an agent $u$ with opinion $\ell$ $\not\in \{\ell_j,j+1\}$. If $\ell>j+1$ $\Challenger{u} = \Defender{u} =\tfalse$ follows from the initialization phase (see \cref{alg:init}). Now assume that $\ell<j+1 $. If opinion $\ell>1$  the opinion was $\Challenger{}$ opinion in tournament $\ell-1$. If $\ell=1$ the opinion was the defender in the first tournament. In either case,  $\Challenger{u}$ and $\Defender{u}$ are set to $=\tfalse$  in Line $17-21$ of \cref{alg:setup} or in Line $5$ of \cref{alg:init}.

\paragraph{Statement~(2)} In this proof we assume  w.l.o.g.\ that $x_{\ell_j}(0) \geq x_{j+1}(0)$. 
Fix a $\Collector$-agent $u$. From Statement~(1) and Statement~(4) of the previous tournament it follows that in Line $5$ of \cref{alg:setup} $\ell[u]$ is set to $\Tokens{u}$ if $u$ is a defender agent and
 to $-\Tokens{u}$ if $u$ is a challenger agent. In Line $8$ of \cref{alg:setup} the defender and challenger agents perform a load balancing protocol for the rest of Phase~2.
The protocol is analyzed in \cite{DBLP:conf/ipps/BerenbrinkFKK19, DBLP:journals/jap/MocquardRSA21}).
We define $L= \sum_{u:\Collector} \ell[u]$ as the total load at the end of Phase~2 and $\hat{L} = \sum_{u:\Collector}|\ell[u]| $ as the total remaining load .
From \cite{DBLP:conf/ipps/BerenbrinkFKK19, DBLP:journals/jap/MocquardRSA21} it follows that at the end of Phase~2 we have  \whp (a) $L = x_{\ell_j}(0) - x_{j+1}(0)$  and (b)
 for every $\Collector$-agent $u$ it holds either $\ell[u] \in \{ 0,1,2 \}$ if $x_{\ell_j}(0) - x_{j+1}(0) \geq |C|/2$  or $\ell[u] \in \{ -1, 0,1 \}$, otherwise.

In Phase~2 of \cref{alg:setup} every $\Collector$-agent $u$ recruits $|\ell[u]|$ many undecided $\Player$-agents $v$. 
If $\ell[u] > 0$ it sets $\Playeropinion{v}=A$ and $\ell[u] = \ell[u] -1$.
If $\ell[u] < 0$ it sets $\Playeropinion{v}=B$ and $\ell[u] = \ell[u] +1$.
For rest of the $\Player$-agents it remains $\Playeropinion{v}=U$.
This is done in Lines $10-12$ of \cref{alg:setup}.
It remains to show that each of these agents can recruit the sufficient amount of $\Player$-agents.
We will show the following claim. 
\begin{claim}
Assume $|P|$ is the number of $\Player$-agents. Fix the configuration at time the end of Phase~2. \Whp we have either 
\begin{itemize}[nosep]
    \item[(i)]  $\hat{L} \leq |P|/2$ , or
    \item[(ii)] for every $\Collector$-agent $u$ we have $\ell[u] \in \{ 0,1,2 \}$ and then there exists a $\Collector$-agent $u$ with $\ell[u]> 0$.
\end{itemize}
\end{claim}
\begin{proof}
Statement~(i) follows directly for $x_{\ell_j}+x_{j+1}\leq |P|/2$. 
Hence, for the rest of the proof we can assume that
$\hat{L}> |P|/2$.  (Note that (i) would immediate follow if $|P|\ge 2 |C|$ which is, unfortunately, quite unlikely).
From the analysis in \cite{DBLP:conf/ipps/BerenbrinkFKK19} it follows that 
\whp at least $|C|/4$ agents $u$ have $\ell[u]=0$ (this holds due to the length of the phase and the fact that in Line $8$ of \cref{alg:setup}  \enquote{$+1$}-s are canceled against \enquote{$-1$}-s).
From Chernoff bounds it follows that \whp $|P|\ge |C|/2$.
Statement~(ii) follows directly for $x_{\ell_j}(0) - x_{j+1}(0) \geq |C|/2$ and the fact that $L = x_{\ell_j}(0) - x_{j+1}(0)$.
The claim follows from a union bound over both statements.
\end{proof}
At last it remains to show that Statement~(2) follows by the claim and the fact that $L = x_{\ell_j}(0) - x_{j+1}(0)$.
Assume Statement~(i) holds. Chernoff bounds show that every $\Collector$-agent $u$ is able to recruit $|\ell[u]|$ many $\Player$-agents  in $\ldauOmicron{n\log n}$ interactions \whp.

Now assume Statement~(ii) holds instead. 
That is, no $\Player$-agent $u$ is able to sets $\Playeropinion{u} = B$ in Line~$11$ of \cref{alg:setup} and hence, it is sufficient that at least a $\Collector$-agent $u$ with $\ell[u] > 0$ is able to recruit a $\Player$-agent.
Again, this follows from Chernoff bounds for $\ldauOmicron{n\log n}$ interactions \whp.
Then  Statement~(2) follows from a union bound.
\paragraph{Statement~(3)}
We  execute the exact majority protocol from~\cite{DBLP:conf/focs/DotyEGSUS21} among the $\Player$-agents in Phase~6.
(The detailed explanation can be found at the beginning of \cref{obs:apply-majority-protocol}.) 
Since the $\Player$ size is at least $n/10$, Chernoff bounds provide sufficiently many meaningful interactions in $\ldauTheta{n\log n}$ interactions \whp.
Together with Statement~(2) this implies the claim.

\paragraph{Statement~(4)}
Similarly to the proof of Statement~(1) we can argue  that in Phase~8 of tournament $j$ each agent interacts at least twice with a $\Player$-agent. From Statement~(3) it follows that every $\Player$-agent $v$ has $\Playeropinion{v}=A$ ($\Playeropinion{v}=B$, respectively) if the defender (challenger, respectively) opinion won the majority protocol in Phase~6. Note  that the competition in the $j$-th tournament is between opinion $j+1$ and $\ell_j$.

First let us assume that for each $\Player$-agent $v$ we have $\Playeropinion{v}=B$, i.e., the challenger opinion won. Consider $\Collector$-agent $u$.  
In Phase~8 $\cref{alg:setup}$ sets (see Line 17-19)
$\Defender{u} =\Challenger{u}$, i.e.,  every $\Collector$-agent $u$ with the challenger opinion has $ \Defender{u}=\ttrue$, and afterwards \cref{alg:setup} sets $\Challenger{u}=\tfalse$.

Now we assume that for each $\Player$-agent $v$ we have $\Playeropinion{v}\in \{A,U\}$.
In that case the defender opinion won the competition and we have  $ \Defender{u}=\ttrue$ for all $\Collector$-agents with the defender opinion, as before. 
\end{proof}

\begin{proof}[Proof of Runtime for Statement~(1) of \cref{thm:result-order}]
We first apply \cref{lem:preprocessing}.
Then it holds  that $t_0(1) = \ldauOmicron{n\cdot (k+\log n)}$ and the population is partitioned into the  roles $\Player, \Tracker, \Clock$ and $\Collector$ where each role consists of at least $n/10$ agents \whp.
The $\Clock$-agents run the phase clock from \cite{DBLP:conf/soda/AlistarhAG18} that provides synchronized phases of length $\ldauTheta{n\log n}$ \whp.
In particular, the separation between even phases is sufficiently large, i.e., it last longer than the time to broadcast a message via one-way epidemic (see \cite{DBLP:journals/dc/AngluinAE08a}).
Now we do an induction over $k-1$ tournaments in order to show that opinion $\ell_k$ is the defender at the end of the tournaments.
At the beginning of the first tournament at time $t_0(1)$ \cref{lem:preprocessing} implies  that opinion $1$ is the initial defender \whp, i.e., $\ell_1 = 1$.
The induction step from tournament $j$ to $j+1$ follows by \cref{lem:setup} \whp.
Thus, the initial plurality opinion is the defender at the end of the last tournament \whp.
At last all agents agree on the unique defender opinion which follows by a final broadcast in $\ldauOmicron{n\log n}$ interactions \whp.
Summing up over the initialization phase and all tournaments, \AlgSimple requires $\ldauOmicron{n\cdot k\cdot \log n}$ interactions in total.
\end{proof}

\section{Removing the Order}%
\label{sec:removing-the-order}

In this section we explain how to remove the assumption that there is an order among the $k$ opinions.
Recall that in \AlgSimple we let opinion $1$ be the first defender and opinion $i+1$ be the challenger of the $i$-th tournament.
The number of tournaments was counted in the $\Tc{}$ variable of $\Tracker$-agents.
Instead, we now assigning the \Tracker agent a slightly different task, and we use a unique leader agent (from the set of \Tracker agents) that randomly \emph{samples} the next challenger before each match.

The leader agent interacts until it encounters an opinion that has not yet participated in a tournaments.
Then the leader agent informs all $\Collector$-agents $u$ with that opinion that they are the next challenger.
Unfortunately, this cannot be done efficiently for each opinion: if $x_j = o(n)$ for some opinion $j$, it takes too long for the leader to interact with an agent of that opinion.
To solve this we use the remaining $\Tracker$-agents.
\Tracker agents copy opinions that have not yet competed in a tournament (using the same number of states as for the counter \Tc{} before).
This effectively amplifies the number of agents having an opinion that has not yet participated in a tournament, making this opinion \emph{visible} to the leader agent.

\paragraph{Challenger and Defender Selection}
Assume for now that we have a unique leader agent.
At the beginning of each tournament $\ell$ in Phase~0 the leader agent and $\Tracker$-agents sample until they meet a $\Collector$-agent with an opinion $j$ that has not yet participated in a tournaments.
As soon as the leader agent has sampled such an opinion $j$ (either from a $\Collector$-agent directly or from a $\Tracker$-agent) it starts to broadcast among the $\Tracker$-agents and the $\Collector$-agents that opinion $j$ is the challenger of tournament $\ell$. (This broadcast is done on a constant fraction of all agents and thus concludes \whp within one phase.)
Now when a $\Collector$-agent $u$ with $\Opinion{u}=j$ interacts with an agent $v$ that knows the challenger opinion, it sets $\Challenger{u} \gets \ttrue$ and becomes a challenger agent.
Note that we can implement this broadcast using one additional bit in the state space. 
Note that we can use the same procedure to select the initial defender before the tournament starts.

\begin{lemma} 
\label{lem:challenger-defender-selection}
Assume a unique leader agent exists. 
Then a challenger (defender) opinion is selected in $\ldauOmicron{n\log n}$ interactions \whp. 
\end{lemma}
\begin{proof}
The lemma follows essentially from the following observation. Let $u$ be an arbitrary but fixed agent and let $A$ be a set of agents with $|A| = \Omega(n)$. Then it follows from Chernoff bounds that in $O(n \log n)$ interactions $u$ interacts with an agent $v \in A$ at least once.

We now give the detailed proof for the correctness of the challenger selection. The defender selection follows by the same arguments.
We call an opinion $j$ \emph{remaining challenger candidate} if the opinion has not participated in a tournament yet.
First we show that the leader agent  selects one of the remaining challenger candidates  in $\ldauOmicron{n\log n}$ interactions \whp.
Then we show that every $\Collector$-agent with opinion $j$ sets its challenger bit in $\ldauOmicron{n\log n}$ interactions \whp.

Let agent $w$ be the leader and let $R$ be the set of agents whose opinions are among the remaining challenger candidates.
If $|R| \geq n/10$ then the probability that the leader $w$ interacts in a fixed step with an agent $v\in R$ is at least constant.
It follows from Chernoff bounds that the leader agent selects a challenger candidate in $\ldauOmicron{n\log n}$ interactions \whp.
Assume $|R| < n/10$.
In this case we first argue that every $\Tracker$-agent $u$ stores the opinion of one the remaining challenger candidates \whp.
This follows from the one-way epidemic process \cite{DBLP:journals/dc/AngluinAE08a} where $R$ is the set of infected agents and the $\Tracker$-agents are susceptible. 
By \cref{lem:preprocessing} it follows that the number of $\Tracker$-agents is at least $n/10$ \whp and hence, the first claim holds.

For the second claim we argue in a similar way.
From the first claim we know that the leader has chosen a challenger opinion $j$ \whp.
The one-way epidemic provides that every $\Tracker$-agent learns the identity of opinion $j$ within $O(n \log n)$ interactions \whp.
Note that we utilize an additional Boolean flag to determine whether a $\Tracker$-agent has already stored the challenger opinion $j$.
It now remains to show every $\Collector$-agent with opinion $j$ interacts at least once with a $\Tracker$-agent \whp.
Again this follows from Chernoff bounds and hence, the second claim holds.

The statement follows from a union bound over both claims.
\end{proof}

Regarding the leader agent we use the leader election protocol from \cite{DBLP:journals/jacm/GasieniecS21} with the phase clock from \cite{DBLP:journals/dc/BerenbrinkEFKKR21}.
We run this protocol among the \Tracker agents.
It requires $\ldauOmicron{\log\log n}$ states and computes a unique leader agent in $\ldauOmicron{n\log^2 n}$ interactions \whp.
Note that the unique leader recognizes when the leader election protocol is concluded. 
This allows us to reuse the states from leader-election and integrate the leader election protocol in an additional, special phase before the tournaments start.

We now describe how we modify \AlgSimple.
The leader election protocol from \cite{DBLP:journals/jacm/GasieniecS21} determines a unique leader agent as follows.
We execute this protocol among the $\Tracker$-agents in a special phase as part of the preprocessing before the first tournament starts.
When a unique leader is elected (\whp), it broadcasts the end of the leader election and initiates the initial defender selection.
The $\Clock$-agents wait in $\Phase{}$ 0 until they receive the signal that a unique leader exists.
The challenger selection is executed at the beginning of a tournament $j$ in Phase~0 and replaces the original challenger selection of \AlgSimple in Lines~2-3 of \cref{alg:setup}.

\begin{proof}[Proof of Statement~(2) of \cref{thm:result-order}]
The result mostly follows from the correctness of \AlgSimple (Statement~(1) of \cref{thm:result-order}).
Again by \cref{lem:preprocessing} it holds that the population is partitioned into the roles $\Player, \Tracker, \Clock$ and $\Collector$, where each role consists of at least $n/10$ agents \whp.
The key modification affects the selection of a unique leader agent, the initial defender opinion, and the challenger opinion for each tournament.

Since the number of $\Tracker$-agents is at least $n/10$, the unique leader agent is computed in $\ldauOmicron{n\log^2{n}}$ interactions \whp by the leader election protocol from \cite{DBLP:journals/jacm/GasieniecS21}.
Then by \cref{lem:challenger-defender-selection} it follows that we have a defender opinion at the beginning of the first tournament \whp.
Similarly to the proof of Statement~(1) of \cref{lem:setup} we can argue with \cref{lem:challenger-defender-selection} that Statement~(1) holds.
It remains to show that the number of states is at most $\ldauOmicron{k + \log n}$. 
The $\Tracker$-agents require $\ldauOmicron{\log\log n}$ states to execute the leader election protocol. 
Until the end of this protocol they do not store any other values. Once the leader election has concluded, they disregard the $O(\log\log n)$ states used for that protocol and use their states to store an opinion instead. Hence, $\ldauOmicron{k + \log n}$ many states are sufficient.

The overall state complexity follows from the proof of Statement~(1) of \cref{thm:result-order} along with the observation that the broadcasts can be implemented using constantly many additional bits.
\end{proof}

\section{Extending the algorithms to large values of k }
\label{sec:appendix-larger-k}

In this section we sketch how \AlgSimple can be adapted to support up to $k<n$ opinions instead of requiring that $k<n/40$.
As a first step, we describe the modification required to support $k < (1-\epsilon) n$, for any arbitrary small constant $\epsilon>0$. The resulting algorithm still respects the number of states and running time stated in \cref{thm:result-order}.
For this, we modify the algorithm at two places: in Line $3$ of \cref{alg:synchronization}, we decrease $\Count u$ by $1/c$ for some large $c$ if $u$ interacts with a collector, and in Line $5$ of \cref{alg:init}, we replace $10$ by some large constant $c'$. This way, there will be a clock agent counting to $5 \log n$ even if only a small (constant) fraction of the collectors has converted its role to a $\Clock$, $\Tracker$ or $\Player$ node. Furthermore, in the chain of inequalities in \cref{claim:lemma1-2}, we can guarantee that 
\[
    \frac{1}{n^2} \sum_{i = 1}^{k} \frac{z_i(t)}{n} 
    \cdot \frac{z_i(t)-1}{n-1}
    = \Omega(\frac{1}{k}).
\]
for $t < \tau(1-\epsilon')$, where $\epsilon'$ is some properly chosen small constant.
Then, we obtain the following adapted statements of  
\cref{lem:preprocessing}:
\begin{enumerate}[nosep]
      \item $\hat{t} = \ldauOmicron{n\cdot (k + \log n)}$.
      \item After interaction $\hat{t}$ each of the roles $\Collector,  \Clock, \Tracker$, and $\Player$ are held by at least $n/C$ agents, for some constant $C$.
      \item After interaction $\hat{t}$ all $\Collector$-agents of opinion $1$ have their defender bit set.
\end{enumerate}
Adapting now the constants in the proof of \cref{thm:result-order} accordingly, we obtain the result for $k < (1-\epsilon) n$, where $\epsilon>0$ can be an arbitrarily small constant.

In order to guarantee the result for any $k < n$, we have to further modify the algorithm. That is, whenever two agents $u$ and $v$ interact, and each of them has exactly one token of the same opinion, then one of them will have two tokens of the same opinion, and the other loses its opinion and becomes a so-called counting agent. The rest of state transitions remain the same for all $\Collector$, $\Clock$, $\Tracker$, and $\Player$-agents being in phase $-1$.  The counting agents start with a counter set to $0$, and every such agent increments its counter if it initiates an interaction with itself (this event occurs with probability $1/n$). If the counter of such an agent hits $C \log n$ for some large $C$, then it broadcasts a message to all agents in the system by converting the phase of any agent to $0$. Additionally, any $\Collector$-agent $u$, which has not interacted with any other agent having the same opinion as $u$, will lose its opinion and is converted to a $\Clock$, $\Tracker$, and $\Player$-agent, each with probability $1/3$. The counting agents convert their role to $\Clock$, $\Tracker$, and $\Player$-agents, each with probability $1/3$. If, however, a $\Clock$-agent switches to phase $0$ first and converts all other agents to phase $0$ as well, then all counting agents will switch their role to $\Clock$, $\Tracker$, and $\Player$-agents with probability $1/3$ each. 
After an agent enters phase $0$, it follows the transitions given in the algorithms in \cref{sec:algorithm}.

In order to argue that this modified algorithm selects the initial majority \whp, we consider three cases. First, assume that $k< n/40$. We know that the number of counting agents can be at most $n/2$. Assume that after $\log^2 n$ rounds, there are no $\Clock$-agents in the system. However, there are at least $n/2$ non-counting agents, and out of these agents, $n/2-\log^2 n$ must be $\Collector$-agents, as otherwise, at least one $\Collector$-agent has converted its role to $\Clock$-agent \whp. Now, we assign each $\Collector$-agent $u$ a matching $\Collector$-agent $v$ having the same opinion as $u$. As there are at most $n/40$ different opinions, we can create $(n/2-log^2 n-n/40)/2$ matching pairs. Since within $n$ interactions, an agent selects another agent with constant probability, at least one $\Collector$-agent selects its matching $\Collector$-agent with constant probability. Thus, in $n$ consecutive interactions, a $\Clock$-agent is created with constant probability. Repeating this argument over $\log n$ rounds consisting of $n$ interactions each, we obtain that \whp a $\Clock$-agent is created. Since we know that an algorithm exists that guarantees the bounds of   
\cref{thm:result-order} for $k < (1-\epsilon) n$, the $\Clock$-agent will trigger phase $0$ after $O(n(k+\log n))$ interactions \whp.

In the second case, we assume that $n/40 \leq k \leq 6n/10$. Here, we consider two subcases. First, assume that there are initially at most $\log^2 n$ agents, which have opinions supported by at least $3$ agents initially. This, however, implies that more than $n/10$ agents have opinions supported by exactly $2$ agents initially. Out of these agents, either at least $n/40$ will convert to counting agents within $O(n \log n)$ rounds or a $\Clock$-agent triggers phase $0$ before that (at that point in time, at least a constant fraction of the initial $\Collector$-agents had to switch their roles). If the first case occurs, then after additional $O(n \log n)$ rounds one of the counting agents triggers phase $0$, and every counting agent is converted to a $\Clock$, $\Tracker$, and $\Player$-agent. Thus, each of these roles is present in a constant fraction of agents \whp. 

In the second subcase, we assume that there are more than 
$\log^2 n$ agents, which have opinions supported by at least $3$ agents initially. Out of these agents, there will be created at least one $\Clock$-agent in $O(n \log n)$ rounds. Since the algorithm ensures that this $\Clock$-agent switches to phase $0$ within additional $O(n)$ rounds, if $k < (1-\epsilon)n$, the claim follows.

In the third case, we assume that $k > 6n/10$. Since there is at least $1$ opinion with support at least $2$, and the number of opinions with support $1$ is at least $n/10$, we obtain that within $O(n \log n)$ rounds at least one counting agent is created. This counting agent triggers phase $0$ after additional $O(n \log n)$ rounds, and all agents with opinions having support $1$ switch their role to $\Clock$, $\Tracker$, and $\Player$-agents. All these roles are supported by at least a constant fraction of the agents \whp. Thus, the algorithm will select the initial majority \whp. 




\section{Auxiliary Results}
\label{apx:auxiliary-results}

In this appendix we present known results used in our analysis.
We start with classical Chernoff bounds.
\begin{theorem}[{\cite[Theorems~4.4,~4.5]{DBLP:books/daglib/0012859}}]
\label{lemma:chernoff_poisson_trials}
\label{thm:chernoff-bounds}
    Let $X_1, \dots, X_n$ be independent Poisson trials with $\Pr[X_i = 1] = p_i$ and let $X = \sum X_i$ with $\Ex{X} = \mu$. Then the following Chernoff bounds hold
    for $0 < \delta \leq 1$:
    \begin{align*}
        \Pr[X > (1+\delta) \cdot\mu] &\leq e^{-\mu \cdot {\delta}^2 / 3}, \quad \text{and} \\
        \Pr[X < (1-\delta) \cdot\mu] &\leq e^{-\mu \cdot {\delta}^2 / 2}.
    \end{align*}
\end{theorem}

\noindent Next we consider tail bounds for sums of geometrically distributed random variables.
\begin{theorem}[{\cite[Theorem~2.1]{arXiv:1709.08157}}]
\label{lem:janson}
Let $X = \sum_{i=1}^{n} X_i$ where $X_i, i = 1,\dots,n$, are independent geometric random variables with $X_i \sim Geo(p_i)$ for $p_i \in (0,1]$.
For any $\lambda \geq 1$,
\begin{equation*}
    \Prob{X \geq \lambda \cdot \Ex{X}} \leq \exp({- \min_{i}\{p_i\} \cdot \Ex{X} \cdot (\lambda - 1 - \ln{\lambda})}).
\end{equation*}
\end{theorem}

The following statement bounds the hitting time for biased random walks. Similar results have already been shown, e.g., in \cite{DBLP:journals/siamcomp/BerenbrinkCSV06,book/mixingtimes,book:Feller-probability}.
For convenience, we give here a combined version of these standard results that fits our needs.
\begin{lemma}[Random Walk Hitting Time]
\label{lem:random-walk}
Let $(W_t)_{t \in \mathbb{N}}$ be a biased random walk on state space $\mathbb{N}_0$, initially at $0$. Let $0 < p < 1$ denote the probability for the walk to move to the right (increase its current position by 1). Conversely, let $q=(1-p)$ denote the probability that it moves to the left (or stays in position in case it currently resides at position $0$). Then, for any $N>0$ and hitting time $\tau_N = \min\{t ~|~ W_t = N\}$ the following holds:
\begin{enumerate}[nosep]
    \item If $p>q$ then $\tau_N \leq (\frac{2}{p-q})^2 \cdot N$ with probability at least $1-\exp(-N)$.
    \item If $p<q$ then $\tau_N \geq (q/p)^{N/2}$ with probability at least $1-(p/q)^{N/2}$. 
\end{enumerate}
\end{lemma}
\begin{proof}
    We start with the first statement and assume $p>q$. We use a similar idea as in Lemma 3.3 of \cite{DBLP:journals/siamcomp/BerenbrinkCSV06}. That is, we let $X_i$ denote a random variable with $X_i = -1$ if the random walk moves to the left, and $X_i = 1$ if it moves to the right in step $i$. Observe that $S_m = \sum_{i=1}^{m} X_i$ minorizes the position $W_m$ of the random walk for any $m \geq 0$. We set $m= (2 /p-q)^2 N$ and apply Hoeffding's bound (Theorem 4.12 of \cite{DBLP:books/daglib/0012859}).  As $-1 \leq X_i \leq 1$ this yields for any $t \geq 0$ that 
    \[
        \Pr \left[ S_m \leq \Ex{S_m} - t \right] \leq \exp(-2t^2 / 4m).
    \]
    Setting $t=\Ex{S_m} - N= m (p-q) - N \geq 0$ this yields that
    \begin{align*}
        \Pr \left[ S_m \leq N \right] &\leq \exp\left(- \frac{(m(p-q) - N)^2}{2m} \right) = \exp\left(-\frac{m(p-q)^2}{2} + N(p-q) - \frac{N^2}{2m} \right) \\
        & = \exp\left(-2N + N(p-q) - \frac{N^2}{2m}\right) \leq \exp(-N).
    \end{align*}
    As  $S_m$ minorizes $W_m$, this implies that the random walk must have hit $N$ before step $m$ with probability at least $1 - \exp(-N)$.
    
    In order to show the second statement, we assume $q<p$ and couple our process with a sequence of gamblers ruin instances. The gambler starts with $1$ money and repeatedly gambles: either it wins $1$ money with probability $p$ or loses $1$ money with probability $q$. The gambler continues until it either runs out of money or reaches a budget of $N+1$. Assume our random walk currently resides at position $0$. We couple its next moves with a gamblers ruin process as follows: the random walk moves to the right each time the gambler wins, otherwise it moves to the left. If the gambler reaches budget $N+1$, then this implies that the random walk hit $N$ before going back to $0$. Otherwise, the gambler runs broke which implies that the random walk is again back at $0$. Hence, we may lower-bound $\tau_N$ by the number of gamblers ruin instances required for the gambler to hit $N+1$ for the first time. According to \cite{book:Feller-probability}  the player reaches the desired budget with probability
    \[
        \frac{\frac{q}{p} - 1}{(\frac{q}{p})^{N+1} - 1} \leq (\frac{p}{q})^{N}.
    \]
    We now apply union bounds, which implies that the gambler wins in any of the first $(q/p)^{N/2}$ instances with probability at most
    \[
        (q/p)^{N/2} \cdot (\frac{p}{q})^{N} = \frac{p}{q}^{N/2}. \qedhere
    \]
    Each gamblers ruin instance corresponds to at least one move of the random walk. Therefore, the number of required gamblers ruin instances serves as a lower bound for the hitting time.
\end{proof}

\end{document}